\documentclass[twocolumn,prl,superscriptaddress, 
]{revtex4-1}
\usepackage{graphicx}
\usepackage{amsfonts}
\usepackage{amssymb}
\usepackage{amsthm}
\usepackage{array}
\usepackage{amsmath}
\usepackage{verbatim} 
\usepackage{hyperref}
\usepackage{color}
\usepackage{bbold}
\usepackage{epstopdf}
\usepackage{mathtools}
\usepackage{enumerate}
\usepackage{ulem}

\newtheorem{proposition}{Proposition}
\newtheorem*{proposition*}{Proposition}

\newtheorem{theorem}{Theorem}
\newtheorem*{theorem*}{Theorem}
\newtheorem{corollary}{Corollary}[theorem]
\newtheorem*{corollary*}{Corollary}

\newtheorem{lemma}{Lemma}

\newcommand{\ket}[1]{\left\vert#1\right\rangle}
\newcommand{\bra}[1]{\left\langle#1\right\vert}

\def\bra#1{\langle #1|}
\def\ket#1{\left|#1 \right>}

\def\Tr{\mbox{Tr}}

\begin{document}
\title{Nonclassicality as a Quantifiable Resource for Quantum Metrology}
\author{Hyukjoon Kwon}
\affiliation{Center for Macroscopic Quantum Control and Institute of Applied Physics, Department of Physics and Astronomy, Seoul National University, Seoul, 151-742, Korea}
\affiliation{QOLS, Blackett Laboratory, Imperial College London, London SW7 2AZ, United Kingdom}
\author{Kok Chuan Tan}
\affiliation{Center for Macroscopic Quantum Control and Institute of Applied Physics, Department of Physics and Astronomy, Seoul National University, Seoul, 151-742, Korea}
\author{Tyler Volkoff}
\affiliation{Department of Physics, Konkuk University, Seoul 05029, Korea}
\author{Hyunseok Jeong}
%\email{h.jeong37@gmail.com}
\affiliation{Center for Macroscopic Quantum Control and Institute of Applied Physics, Department of Physics and Astronomy, Seoul National University, Seoul, 151-742, Korea}
\date{\today}
\begin{abstract}
We establish the nonclassicality of continuous-variable states as a resource for quantum metrology. Based on the quantum Fisher information of multimode quadratures, we introduce the metrological power as a measure of nonclassicality with a concrete operational meaning of displacement sensitivity beyond the classical limit. This measure belongs to the resource theory of nonclassicality, which is nonincreasing under linear optical elements. Our Letter reveals that a single copy, highly nonclassical quantum state is intrinsically advantageous when compared to multiple copies of a quantum state with moderate nonclassicality.  This suggests that metrological power is related to the degree of quantum macroscopicity. Finally, we demonstrate that metrological resources useful for nonclassical displacement sensing tasks can be always converted into a useful resource state for phase sensitivity beyond the classical limit.
\end{abstract}
\maketitle

Recognizing the differences between classical and quantum physics has changed our viewpoint of nature, while developments in quantum information theory have shown that these differences lead to quantum advantages in informational tasks~\cite{Ekert91, Bennett92, Bennett1993, Giovannetti06, Gisin07, Pironio2010, Raussendorf2001}. Nonclassicality can be defined by the negativity of the Glauber-Sudarshan $P$ representation in the context of  light fields~\cite{Glauber63, Sudarshan63, Titulaer1965}. An $N$-mode continuous-variable state $\hat\rho$ can be represented as
$$
\hat\rho = \frac{1}{\pi^N} \int d^{2N} \boldsymbol{\alpha} P_{\hat\rho} (\boldsymbol\alpha) \ket{\boldsymbol\alpha} \bra{\boldsymbol\alpha},
$$
where $P_{\hat\rho}(\boldsymbol\alpha)$ is the $P$ function and the set of coherent states forms an overcomplete basis $\ket{\boldsymbol\alpha} = \bigotimes_{n=1}^N\ket{\alpha_n}$ in the corresponding Hilbert space. As coherent states are considered to be the most classical states among all pure states \cite{Glauber63, Sudarshan63, Mandel86, Zurek93}, a nonclassical quantum state, which cannot be represented as a statistical mixture of coherent states, should contain negativity in its $P$ function  \cite{Titulaer1965}.

A diverse range of studies have been performed to characterize nonclassicality \cite{Vogel00, Richter02, KZ04, Jiang13, Rigovacca16, Perina17, Damanet18}, as well as its relationship to entanglement \cite{Kim02, Asboth05, Vogel14} and quantum communications \cite{Smith11, Lercher13}. For the quantification of nonclassicality, various approaches have been suggested including distance-based measures \cite{Hillery1987, Marian2002}, nonclassicality depth \cite{Lee91}, entanglement potential \cite{Asboth05, Vogel14}, characteristic function methods \cite{Ryl17}, and operational approaches \cite{Rivas10, Gehrke12}.  
Recently, the nonclassicality based on the negativity of the $P$ function was investigated   using the resource
theory of coherence \cite{Bobby17}. 
The orthogonalization process suggested in Ref.~\cite{Bobby17} successfully unifies the old notion of nonclassicality \cite{Glauber63, Sudarshan63, Titulaer1965} and the new concept of  coherence \cite{Baumgratz14} in the coherent-state basis. Emerging from this characterization is a resource theory of nonclassicality based on linear optics, where the set of classical operations are naturally chosen as linear optical operations. The challenge is then to find a quantifier of nonclassicality based on the resource theory that possesses a clear operational significance, paralleling the developments in the entanglement \cite{Horodecki09} and coherence \cite{Baumgratz14, StreltsovRMP} theories. It has been found  that in metrological tasks, nonclassicality rather than entanglement is a necessary resource to achieve quantum advantages \cite{Sahota15, Friis15, Ge18}, while the operational meaning of nonclassicality was very recently studied based on the quadrature fluctuations in a similar vein \cite{YadinX,noteadded}.

In this Letter, we demonstrate that the nonclassicality of a continuous-variable state is a quantifiable resource for parameter estimation tasks.
We show that the mean quadrature variance captures every pure-state nonclassicality, and its convex roof construction becomes a strict measure of nonclassicality. Extending this concept, we introduce the metrological power to quantify nonclassical resources that lead to quantum enhancement in displacement metrology, given in the form of the quantum Fisher information (QFI). We prove that this quantifier witnesses the negativity of the $P$ function and does not increase by linear optical operations so that it belongs to the family of monotones within the resource theory of nonclassicality. In addition, it is shown that a collection of many small-size nonclassical states cannot achieve a large degree of nonclassicality; this is consistent with the notion of quantum macroscopicity \cite{Leggett80,review2015,Frowis18}. Interestingly, nonclassical resources for displacement sensing can always be converted into a useful resource for phase sensing tasks using linear optical operations. Our Letter provides a concrete operational meaning for the nonclassicality of a continuous-variable state as a potential resource for quantum metrology that can be quantified by a computable measure.

{\it Resource theory of nonclassicality.---} 
We first define a resource theory of nonclassicality based on Ref.~\cite{Bobby17}. Consider a linear optical unitary for the $N$-mode bosonic system belonging to the $O(2N)$ rotation group of the quadratures $\hat{\boldsymbol{R}}:=(\hat{x}_1, \hat{p}_1, \cdots, \hat{x}_N, \hat{p}_N)^T$, in addition to the displacement operation $\hat{D}_n(\alpha_n) = \exp[\alpha_n \hat{a}_n^\dagger - \alpha_n^* \hat{a}_n]$. Such a unitary transforms a multimode bosonic operator $\hat{a}_{\boldsymbol{\mu}}^\dagger := \sum_{n=1}^N \mu_n \hat{a}_n^\dagger$ into $\hat{a}_{\boldsymbol\mu'}^\dagger +  \bigoplus_{n=1}^N \alpha_n {\mathbb 1}_n$, where $\boldsymbol\mu := ({\rm Re}[\mu_1], {\rm Im}[\mu_1], {\rm Re}[\mu_2], {\rm Im}[\mu_2], \cdots, {\rm Re}[\mu_N], {\rm Im}[\mu_N])^T$ is a real $2N$-dimensional unit vector %satisfying $\sum_{n=1}^N |\mu_n|^2 = 1$,
 and $\alpha_n \openone_n$ corresponds to the displacement on the $n$th mode.
 % with a complex amplitude $\alpha_n$.
 Consequently, a mutimode quadrature operator can be defined as $\hat{X}_{\boldsymbol\mu} := (\hat{a}_{\boldsymbol{\mu}} + \hat{a}_{\boldsymbol{\mu}}^\dagger)/\sqrt{2} = \hat{\boldsymbol R}^T \boldsymbol{\mu}$. 
\begin{figure}[t]
\includegraphics[width=0.6\linewidth]{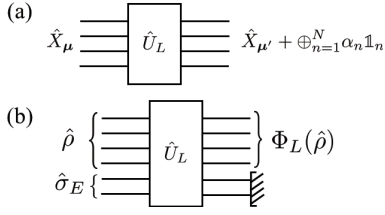}
\caption{(a) Linear optical unitary and (b) linear optical map.}
\label{FIG1}
\end{figure}\
Using linear optical unitary operations, we define a linear optical map
$$\Phi_L({\hat\rho}_A) := \mathrm{Tr}_E [ \hat{U}_L ({\hat\rho}_A \otimes \hat\sigma_E) \hat{U}_L^\dag ],$$ 
where $\hat\sigma_E$ is a classical state (see Fig.~\ref{FIG1}), as a free operation since it  maps every classical state into another classical state. A selective linear operation can be defined by a set of Kraus operators $\{{\hat{K}}_i\}$ when there exists $\hat{U}_L$, classical ancilla $\hat\sigma_{EE'}$, and a set of orthogonal vectors $\{\ket{i}_{E'} \}$ such that $\mathrm{Tr}_E [ \hat{U}_L ({\hat\rho}_A \otimes \hat\sigma_{EE'}) \hat{U}_L^\dag ]= \sum_i p_i {\hat\rho}^i_{A} \otimes \ket{i}_{E'}\bra{i}$, where $p_i {\hat\rho}^i_{A} \coloneqq {\hat{K}}_i{\hat\rho}_A{\hat{K}}^\dag_i $ and $p_i \coloneqq \mathrm{Tr}({\hat{K}}_i{\hat\rho}_A{\hat{K}}^\dag_i)$.
One might expect that a complete set of classicality preserving maps could be expressed in the form of dilations of a linear optical unitary with classical ancilla, but this is not the case. We note that a classicality preserving map $\Lambda: \hat\rho \rightarrow \int \frac{d^2\alpha}{\pi} Q_{\hat\rho}(\alpha) \ket{\alpha} \bra{\alpha}$, where $Q_{\hat\rho}(\alpha) = \bra{\alpha} \hat\rho \ket{\alpha}$ is the Husimi $Q$ function, is not a linear optical map \cite{Diaz18} since it involves a metaplectic unitary corresponding to two-mode squeezing \cite{Palma18}. Nevertheless, a set of linear optical maps serves as an important class of operations that can be relatively easily performed in laboratories, compared to nonlinear operations such as squeezing.

In this framework, nonclassicality for a pure state $\ket{\psi}$ can be quantified by the mean quadrature variance
\begin{equation}
\overline{\cal V}(\ket{\psi}) := \frac{1}{N} \sum_{k=1}^{2N} {\rm Var}(\psi, \hat{R}^{(k)}),
\end{equation}
where ${\rm Var} (\psi, \hat{O}) := \langle \psi | \hat{O}^2 | \psi \rangle -  \langle \psi | \hat{O}| \psi \rangle^2 $ and $\hat{R}^{(k)}$ is the $k$th element of $\hat{\boldsymbol R}$. It is important to note that $\overline{\cal V} \geq 1$, and the equality holds {\it if and only if} the state is a coherent state. We extend this measure to quantify the nonclassicality of a mixed state by taking the convex roof:
\begin{equation}
{\cal Q}(\hat\rho) := \min_{\{p_i, \ket{\psi_i} \}} \sum_i p_i \overline{\cal V}(\ket{\psi_i}) -1,
\end{equation}
where $\{ p_i, \ket{\psi_i} \}$ is a pure-state decomposition of $\hat\rho$.
% satisfying $\hat\rho = \sum_i p_i \ket{\psi_i} \bra{\psi_i}$ with $p_i \geq 0$.
We show that ${\cal Q}$ is a faithful measure of nonclassicality \cite{Bobby17}.
\begin{theorem}
\label{NCM}
${\cal Q}$ is a nonclassicality measure satisfying the following conditions.
\begin{enumerate}
\item $\mathcal{Q}({\hat\rho}) = 0$ {\it  if and only if }${\hat\rho}$ is classical.

\item
	\begin{enumerate}
	
	\item  (Weak monotonicity) $\mathcal{Q}({\hat\rho})\geq \mathcal{Q}(\Phi_L ({\hat\rho}))$.
	
	\item (Strong monotonicity) $\mathcal{Q}({\hat\rho}) \geq \sum_i p_i\mathcal{Q}({\hat\rho}_i)$ where $p_i \coloneqq \mathrm{Tr}({\hat{K}}^\dag_i{\hat{K}}_i {\hat\rho} )$ and ${\hat\rho}_i \coloneqq ({\hat{K}}_i {\hat\rho} {\hat{K}}^\dag_i)/ p_i$.
	
	\end{enumerate}

\item (Convexity), i.e. $\mathcal{Q}(\sum_i p_i {\hat\rho}_i) \leq \sum_i p_i \mathcal{Q}( {\hat\rho}_i)$.
\end{enumerate}
\end{theorem}
We note that the value of nonclassicality is bounded by ${\cal Q}(\hat\rho) \leq 2 (\bar{n}/N)$, where $\bar{n} := {\rm Tr} \left[ \sum_{n=1}^{N} \hat{a}^\dagger_n \hat{a}_n \hat\rho \right]$ is the mean photon number. The upper bound saturates in the case of pure states 
if and only if $\langle \psi \vert \hat{R}^{(k)} \vert \psi \rangle=0$ for every $k$, a condition which holds for, e.g. Fock states, cat states, or squeezed coherent states. Another interesting point is that $N (\bar{\cal V} -1)$ is equivalent to the phase space macroscopicity measure proposed in Ref.~\cite{Lee11}. Thus, ${\cal Q}$ can be understood as the convex roof extension of the macroscopicity measure per mode.

{\it Nonclassicality and metrological power.---}
We now establish the relationship between the quadrature variance and the displacement  sensitivity. Suppose that we want to estimate the parameter $\theta$ when a quantum state $\hat\rho$ is displaced into $\hat\rho_{\theta, \boldsymbol \mu} = e^{- i\theta \hat{X}_{\boldsymbol \mu}} \hat\rho e^{ i\theta \hat{X}_{\boldsymbol \mu}}$.
In this case, a tight bound for the variance of the estimator $(\Delta \theta)_{\boldsymbol \mu}^2$ by performing the optimal measurements on $\hat\rho_{\theta, \boldsymbol \mu}$ is given by the quantum Cram{\'e}r-Rao bound \cite{Braunstein94}
\begin{equation}
\label{CRB}
(\Delta \theta)^2_{\boldsymbol \mu} \geq \frac{1}{I_F(\hat\rho, \hat{X}_{\boldsymbol \mu})} = \frac{1}{\boldsymbol \mu ^T {\boldsymbol F} \boldsymbol \mu}.
\end{equation}
The QFI can be calculated as $I_F(\hat\rho, \hat{X}_{\boldsymbol \mu}) = 2 \sum_{i,j} \frac{(\lambda_i - \lambda_j)^2}{\lambda_i + \lambda_j}  |\bra{i}\hat{X}_{\boldsymbol \mu}\ket{j}|^2$ by using the eigenvalue decomposition $\hat\rho = \sum_i \lambda_i \ket{i}\bra{i}$, and ${\boldsymbol F}$ is the QFI matrix with real symmetric $2N \times 2N$ elements $F  _{kl} = 2 \sum_{i,j} \frac{(\lambda_i - \lambda_j)^2}{\lambda_i + \lambda_j} \bra{i}\hat{R}^{(k)}\ket{j} \bra{j}\hat{R}^{(l)}\ket{i}$. For a pure state four times the variance is equal to the QFI, so that a large quadrature variance directly implies high displacement sensitivity. The QFI has been also studied to quantify multipartite entanglement \cite{Pezze09, Hyllus12, Toth12} and macroscopic quantum coherence \cite{Frowis12, Yadin16, Frowis17, Frowis18}.

From the observation that ${\boldsymbol F}/2 = {\mathbb 1}$ for every coherent state, $(\Delta \theta)^2_{\boldsymbol \mu}$ is lower bounded by $1/2$ for any classical state, so-called the standard quantum limit (SQL) for displacement metrology. The pure-state nonclassicality measure ${\cal Q}(\ket{\psi}) = {{\rm Tr}[{\boldsymbol F}}]/{(4N)} - 1$ has the meaning of the metrological advancement beyond the SQL, on average over all possible values of $\boldsymbol\mu$. For a mixed state, however, it is unknown if ${\cal Q}$ has a direct operational meaning in terms of quantum metrology, while ${\rm Tr} [{\boldsymbol F}]$ cannot fully capture nonclassicality of a mixed state when some eigenvalues of ${\boldsymbol F}$ are smaller than $2$.

Nonetheless, we shall consider another quantifier of nonclassicality, the ``metrological power", which has the concrete operational meaning of the maximal metrological advantage by performing a linear optical unitary with a vacuum ancilla:
\begin{equation}
\label{MQC}
\begin{aligned}
{\cal M} (\hat\rho) &:=   
\frac{1}{2} \max_{\hat\sigma =\hat{U}_L (\hat\rho \otimes \ket{0}\bra{0}) \hat{U}_L^\dagger}  I_F( \hat\sigma, \hat{x}_1) -1 \\
&= \max \left\{ \frac{ \lambda_{\max} ({\boldsymbol F}) }{2} - 1 , 0 \right\},
\end{aligned}
\end{equation}
where $\lambda_{\max} ({\boldsymbol F})$ is the maximum eigenvalue of $\boldsymbol F$. This quantifies the optimal sensitivity among all possible parametrizations since $\displaystyle \min_{\boldsymbol \mu} (\Delta \theta)^2_{\boldsymbol \mu} =  [\lambda_{\rm max}(\boldsymbol F)]^{-1}$, together with the fact that one can always find a linear optical unitary operator $\hat{U}_L$ such that $\hat{U}_L^\dagger e^{- i \theta \hat{X}_{\boldsymbol\mu}} \hat{U}_L = e^{-i \theta \hat{x}_1}$, and displacement operations do not change the quadrature QFI. We show the following useful properties of ${\cal M}$:

\begin{theorem} The metrological power ${\cal M}$ satisfies the following properties:
\label{QFT}
\begin{enumerate}
\item ${\cal M} \geq 0, $ and ${\cal M}=0$ for every classical state. For a pure state, ${\cal M} = 0$ {\it if and only if} the state is a coherent state.
\item ${\cal M}$ is invariant under linear optical unitaries $\hat{U}_L$ and a monotone under linear optical maps $\Phi_L$,
\item ${\cal M}$ is convex,
\item ${\cal M}(\hat\rho_A \otimes \hat\sigma_B) = \max \left\{ {\cal M}(\hat\rho_A), {\cal M}(\hat\sigma_B) \right\}$
\end{enumerate}
\end{theorem}
The first property shows that every pure quantum state except coherent states outperforms all classical states in terms of the metrological power.  For a mixed state, this quantifier can witness nonclassicality whenever ${\cal M}>0$, although there can exist nonclassical states having ${\cal M} = 0$. This is, however, offset by the computational advantages and operational interpretation of ${\cal M}$. The metrological power also satisfies  monotonicity and convexity which are necessary conditions for nonclassicality monotones. The last property fulfills one of the proposed requirements to quantify genuine quantum macroscopicity: the accumulation microscopic quantum coherence should be distinguished from the genuine macroscopic coherence \cite{Leggett80}.

Similar quantum macroscopicity measures for optical systems have been proposed \cite{Oudot15, Frowis17, Volkoff14} based on the QFI, for instance the quantity $\max_{ \{ \phi_n \}} I_F(\hat\rho, \hat{X}_{\{ \phi_n \}}) / N$ using the sum of quadratures $\hat{X}_{\{ \phi_n \}} =  \sum_{n=1}^{N} [\cos\phi_n \hat{x}_n + \sin\phi_n \hat{p}_n ]$. In this case, however, we point out that a linear optical unitary can increase the measure, since $\hat{X}_{\{ \phi_n \}}$ in general does not transform in a covariant way, i.e. $\hat{U}_{L}^\dag \hat{X}_{\{ \phi_n \}} \hat{U}_{L} \neq \hat{X}_{\{ \phi'_n \}}$. Thus, measures of this type do not belong to nonclassicality monotones, although they capture many useful properties of quantum macroscopicity. It is worth noting that utilizing the quadrature QFI to characterize nonclassicality was recently studied with a slightly different set of free operations  \cite{YadinX}.
 
{\it Examples.---}
We first observe that both the Fock state $\ket{n}$ and NOON state $\ket{n}\ket{0} + \ket{0}\ket{n}$ give ${\cal Q} = 2  \bar{n}/N$ and ${\cal M} = 2 \bar{n}$. A cat state $\ket{\alpha} \pm \ket{-\alpha}$ gives ${\cal Q} = 2  \bar{n}$ and ${\cal M} = 2(\bar{n} + |\alpha|^2)$, while a decohered cat state $\hat\rho_{\Gamma} = {N_\Gamma}^{-1} \left[ \ket{\alpha} \bra{\alpha} + \ket{-\alpha} \bra{-\alpha} + \Gamma (\ket{\alpha} \bra{-\alpha} + \ket{-\alpha} \bra{\alpha}) \right]$ gives  ${\cal M}(\hat\rho_\Gamma) = \max\left\{ \frac{16 |\alpha|^2}{N_\Gamma^2} \Gamma (\Gamma + e^{-2|\alpha|^2}), 0 \right\}$, where $N_\Gamma = 2 + 2 \Gamma e^{-2|\alpha|^2}$. A decohered even cat state with positive $\Gamma$ is nonclassical unless $\Gamma = 0$, while a decohered odd cat state with negative $\Gamma$ can be nonclassical when ${\cal M} = 0$. Because of invariance under linear optical unitary operations, nonclassicality between different modes can also be fairly compared throughout our measure. For example, ${\cal M}$ for an entangled coherent state $\ket{\alpha} \ket{\alpha} \pm \ket{-\alpha}  \ket{-\alpha} $ is equivalent to a single-mode cat state with an amplitude $\sqrt{2}\alpha$ since they are interconvertible via a $50:50$ beam splitter. This can be extended to a multimode entangled coherent state $\ket{\alpha_1} \ket{\alpha_2} \cdots \ket{\alpha_N} \pm \ket{-\alpha_1} \ket{-\alpha_2} \cdots \ket{-\alpha_N}$ which is convertible into $( \ket{\gamma} \pm \ket{-\gamma}) \ket{0} \cdots \ket{0}$ via beam splitter operations, where $|\gamma| = \sqrt{\sum_{n=1}^N |\alpha_n|^2}$.

We also apply our result to a multimode Gaussian state characterized by its mean value $\boldsymbol d$ with $d_k = {\rm Tr} \left[ \hat\rho \hat{R}^{(k)} \right]$ and the covariance matrix $\boldsymbol V$ with $V_{kl} = {\rm Tr} \left[ \hat\rho \{ \hat{R}^{(k)} - d_k, \hat{R}^{(l)} - d_l\} \right]$, where $\{\hat{A}, \hat{B}\}:= \hat{A}\hat{B} + \hat{B}\hat{A}$. The symplectic transform of $\boldsymbol{V}$ and corresponding symplectic matrix $\boldsymbol S$ then always exist. This allows us to decompose every Gaussian state into single-mode squeezing combined with linear optical operations acting on the product of thermal states  \cite{Weedbrook12}.
In this case, the following closed form formula is obtained: ${\cal M} = \max \{ \lambda_{\max} \left[\boldsymbol S^{-1} \boldsymbol S^T \boldsymbol V^{-1} \boldsymbol S (\boldsymbol S^{-1})^T \right] - 1 , 0 \}$.
 \begin{figure}[t]
\includegraphics[width=0.95\linewidth]{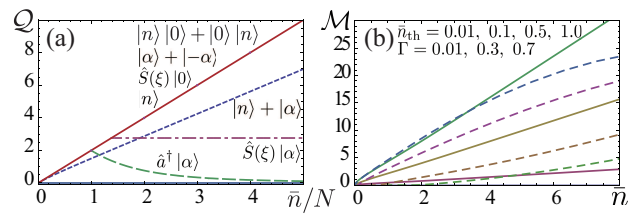}
\caption{(a) Nonclassicality measure ${\cal Q}$ achieves the maximum value (solid line) for NOON,
% $\ket{n}\ket{0}+ \ket{0}\ket{n}$
cat,
% $\ket{\alpha} + \ket{-\alpha}$, 
squeezed,
% $\hat{S}(\xi)\ket{0}$,
and Fock states.
% $\ket{n}$.
Also, superposition between Fock state and coherent state $\ket{n} + \ket{\alpha}$ with $n = |\alpha|^2$ (dotted line), squeezed coherent states $\hat{S}(\xi)\ket{\alpha}$ for $\xi=1$ (dot-dashed line), and photon-added coherent states $\hat{a}^\dagger \ket{\alpha}$ (dashed line) are evaluated. (b) Metrological power ${\cal M}$ for decohered cat states $\hat{\rho}_{\Gamma}$ (solid lines) and squeezed thermal states $\hat{S}(\xi) \hat\tau \hat{S}^\dagger(\xi)$ (dashed lines) with the parameters $\Gamma = 0.01,~0.3,~0.7$ and $\bar{n}_{\rm th} = 0.01,~0.1,~0.5,~1.0$ (both starting from above).}
\label{FIG3}
\end{figure}
Especially for a single-mode Gaussian state $\hat{D}(\alpha) \hat{S}(\xi) \hat\tau \hat{S}^\dagger(\xi) \hat{D}^\dagger(\alpha)$ with $\hat{S}(\xi) = \exp[(\xi \hat{a}^{\dagger 2} - \xi^* \hat{a}^2)/2]$ and $\hat\tau = \sum_{n=0}^\infty \bar{n}_{\rm th}^n /(1+\bar{n}_{\rm th})^{(n+1)} \ket{n}\bra{n}$, a direct relationship between nonclassicality and squeezing \cite{Idel16} can be derived as ${\cal M} = e^{2G(\boldsymbol{V})} - 1 = {\rm max} \left\{ \exp(2|\xi|)/(2 \bar{n}_{\rm th} +1)-1, 0 \right\}$, where $G(\boldsymbol{V}) := \inf \left[ -\sum_{i=1}^N\log s_i^\downarrow (\boldsymbol{S}) | \boldsymbol{V} \geq \boldsymbol{S^T S} \right]$ with $s_i^\downarrow(\boldsymbol{S})$ being singular values of $\boldsymbol{S}$ in decreasing order. This observation also leads to the following corollary:
\begin{corollary}
The metrological power ${\cal M}$ is zero if and only if a single-mode Gaussian state is classical.
\end{corollary}
\noindent Similar to the case of entangled coherent states, the metrological power of two-mode and single-mode squeezed states can be equivalently compared as they are interconvertible by using the beam splitter. Figure~\ref{FIG3} shows ${\cal Q}$ and ${\cal M}$ for various types of quantum states.

{\it Quantum phase estimation assisted by linear optical unitaries.---}
We discuss how a nonclassical resource for the displacement metrology can be utilized in phase estimation tasks beyond the classical limit. Quantum phase estimation aims to measure the relative phase of a chosen mode of a interferometer whose dynamics is given by $e^{-i \theta \hat{a}^\dagger \hat{a}} \hat\rho e^{i \theta \hat{a}^\dagger \hat{a}}$. The sensitivity of the phase estimation task is bounded by $(\Delta \theta)^2_{\rm phase} \geq {I_F(\hat{\rho}, \hat{a}^\dagger\hat{a})}^{-1}$. It was shown \cite{Rivas10} that  a nonclassical quantum state can be identified whenever the QFI is larger than four times of the mean photon number $I_F(\hat\rho, \hat{a}^\dagger\hat{a}) > 4 {\rm Tr} [\hat\rho \hat{a}^\dagger\hat{a}]$, where the SQL for the phase metrology can be considered as $I_F(\hat\rho, \hat{a}^\dagger\hat{a}) \leq 4 {\rm Tr} [\hat\rho \hat{a}^\dagger\hat{a}]$. Although this condition is useful to witness nonclassicality, we highlight that it is not sufficient to detect every nonclassical pure state. For example, the Fock state $\ket{n}$ is obviously nonclassical for $n > 0$, but $I_F(\ket{n}, \hat{a}^\dagger\hat{a} ) = 0$.

In order to overcome this problem, we optimize the phase sensitivity over linear optical unitaries, analogously to the displacement metrology. However, we should additionally take into account that displacement can increase the phase estimation sensitivity even for classical states as $I_F(\ket{\alpha}, \hat{a}^\dagger\hat{a}) \propto |\alpha|^2$ for a coherent state $\ket{\alpha} = \hat{D}(\alpha)\ket{0}$. It is therefore necessary to characterize the linear optical unitaries according to the degree of displacement. This can be done by decomposing a linear optical unitary into $\hat{U}_L^\alpha := \left[ \bigotimes_{n=1}^N \hat{D}_n(\alpha_n) \right] {\hat{U}}_L^0$ with  $|\alpha|^2= \sum_{n=1}^N |\alpha_n|^2 $, and $\hat{U}_L^0$ is a linear linear optical unitary without any displacement. We then define the {\it $\alpha$-invested metrological power} for phase estimation as
\begin{equation}
{\cal M}_{\rm phase}^\alpha (\hat\rho) := \max_{\hat\sigma = \hat{U}_L^\alpha (\hat\rho \otimes \ket{0}\bra{0})  \hat{U}_L^{\alpha \dagger} }\left[ \frac{I_F(\hat{\sigma}, \hat{a}_1^\dagger \hat{a}_1)}{4} -   {\rm Tr} [\hat\sigma  \hat{a}_1^\dagger \hat{a}_1] \right],
\end{equation}
where ${\cal M}_{\rm phase}^\alpha \geq {\cal M}_{\rm phase}^\beta \geq 0$ for $|\alpha| \geq |\beta|$ and ${\cal M}_{\rm phase}^\alpha = 0$ for every classical state, thus the phase sensitivity beyond the SQL (${\cal M}_{\rm phase}^\alpha >0$) directly captures the negativity in the $P$ distribution. Additionally, ${\cal M}_{\rm phase}^\alpha$ enjoys convexity and is invariant under 
%a passive linear optical unitary
$\hat{U}_L^0$. We demonstrate a remarkable relationship between the displacement and phase metrological powers.

\begin{figure}[t]
\includegraphics[width=0.65\linewidth]{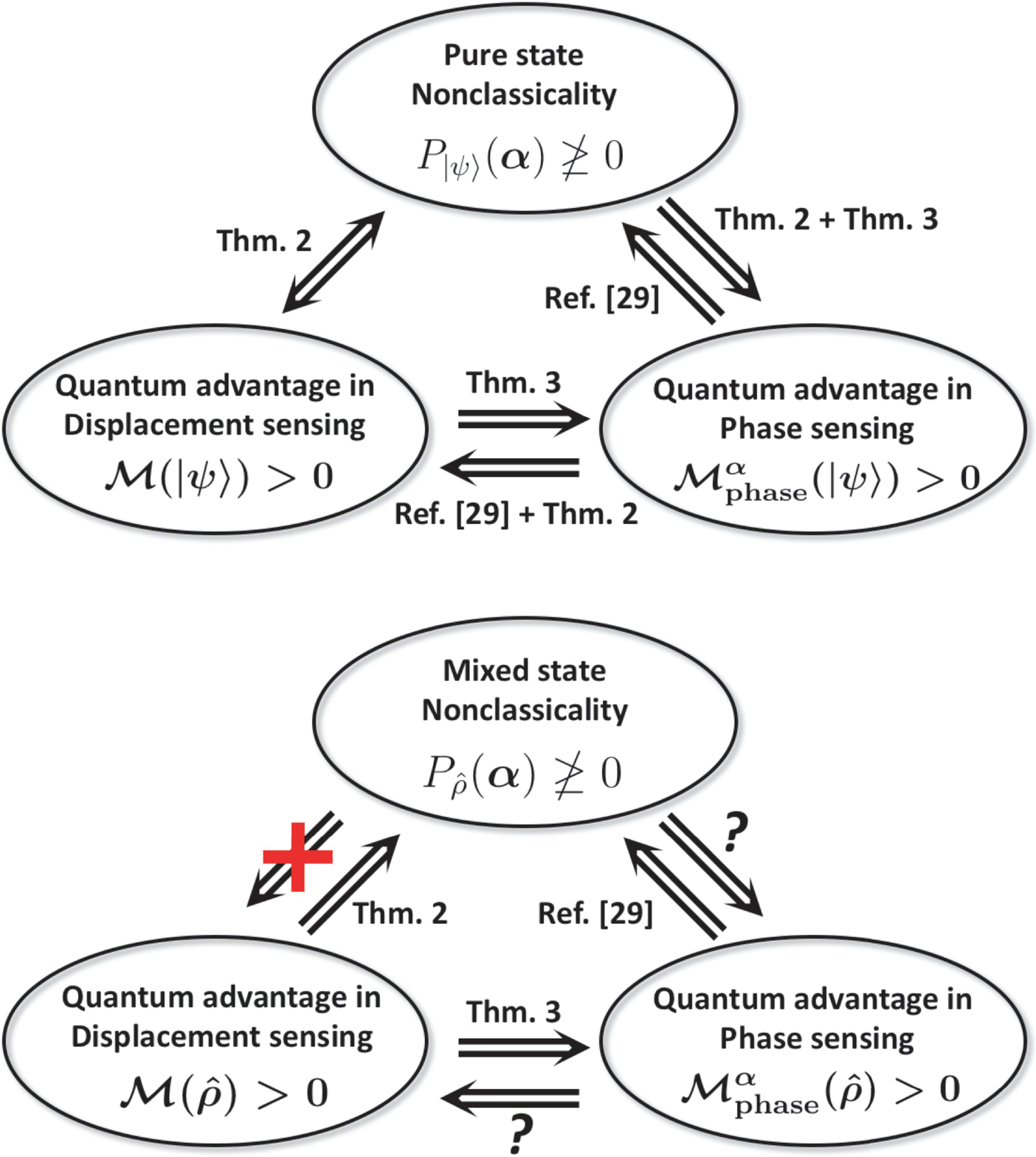}
\caption{Relationship between nonclassical resources for the displacement and phase estimation tasks.}
\label{FIGNEW}
\end{figure}

\begin{theorem} Provided ${\cal M}(\hat\rho) >0$, there exists a linear optical unitary $\hat{U}_L^\alpha$ to reach the sensitivity beyond the SQL for phase estimation, i.e., ${\cal M}_{\rm phase}^\alpha(\hat\rho) >0$.
\label{SQL}
\end{theorem}
\noindent In particular,  ${\cal M} \leq {\displaystyle \lim_{|\alpha| \rightarrow \infty} [{{\cal M}_{\rm phase}^\alpha} /{|\alpha|^2}] }\leq {\cal M} +1$ (\cite{note1}, see also the Supplemental Material \cite{Suppl}),
which can be intuitively understood by the fact that a large displacement followed by a small rotation can be approximated by two sequential displacement operations in orthogonal directions. Another important figure of merit for phase metrology is the scaling behavior with the mean photon number $\bar{n}$. In phase estimation, the classical limit with coherent states is given by $\Delta \theta_{\rm cl} \propto 1/\sqrt{\bar{n}}$, while quantum states can achieve the sensitivity of $\Delta \theta_{\rm HS} \propto 1/\bar{n}$, referred as Heisenberg-like scaling (HS) \cite{Giovannetti11}. In order to reach HS, the corresponding QFI should scale quadratically with $\bar{n}$.
The following Theorem demonstrates that high nonclassicality in displacement sensing is sufficient to achieve HS.
\begin{theorem} If ${\cal M}(\hat\rho) \propto \bar{n}_{\hat\rho}$,
there exists $\hat{\sigma}=\hat{U}_{L}^{\alpha}(\hat\rho \otimes \vert 0 \rangle \langle 0 \vert ) \hat{U}_{L}^{\alpha \dagger}$ such that $I_{F}( \hat{\sigma},\hat{a}_{1}^{\dagger}\hat{a}_{1}) \propto \bar{n}_{\hat\sigma}^2$.  More precisely, HS can be achieved if and only if ${\cal M}(\hat\rho_0) \propto \bar{n}_{\hat\rho_0}$ or ${\cal M}_{\rm phase}^0 (\hat\rho_0) \propto \bar{n}_{\hat\rho_0}^k$ with $k\geq 2$, where $\hat\rho_0= \hat{V}_L \hat\rho \hat{V}_L^\dagger$ is the state centered in phase space (${\rm Tr} \hat\rho_0 \hat{\boldsymbol R} = 0$) by acting the linear optical unitary $\hat{V}_L$ on $\hat\rho$. Here,  $\bar{n}_{\hat\sigma}$ is the mean photon number of a quantum state $\hat\sigma$.
\label{HL}
\end{theorem}
\noindent We note that the Fock state and cat state cannot reach HS via only linear interferometers without additional displacement. However, Theorem~\ref{HL} guarantees that an appropriate displacement operation will allow the system to reach HS. According to Theorems~\ref{SQL} and \ref{HL}, a nonclassical resource for displacement sensing always implies the quantum enhancement of phase sensing. However, it is unclear at this point whether {\it1. negative $P$ function of a mixed state always implies quantum enhancement in phase sensing} and whether {\it 2. nonclassical phase sensing implies nonclassical displacement sensing} (see Fig.~\ref{FIGNEW}). These two statements are incompatible thus both cannot be true at the same time, but both can be false.

{\it Remarks.---}
We have identified nonclassicality of continuous-variable states as a quantifiable resource for quantum metrology. We have shown that any pure state with negativity in the $P$ function provides metrological enhancement over all classical states in displacement estimation tasks, and so does every single-mode Gaussian state. This metrological power is found to be a measure of nonclassicality based on a quantum resource theory that does not increase under linear optical elements. It is demonstrated that every state displaying metrological enhancement in displacement sensing can be converted into nonclassical phase sensitivity by utilizing a linear optical unitary.

The metrological power also satisfies the necessary conditions for a valid measure of quantum macroscopicity. Our study provides a possible avenue to a unified understanding of nonclassicality, quantum macroscopicity, and the metrological usefulness in the framework of the quantum resource theory. Our measures could possibly be applied not only to multimode bosonic systems, but also to other many-body systems including spin, atomic, and optomechanical systems. In these systems, a more general notion of coherent states \cite{Radcliffe71,Drummond16} and the metrological usefulness with nonclassical states \cite{GB2015,FB2017} can be considered. This may lead to a unified description of nonclassicality for both discrete and continuous systems.

\begin{acknowledgments}
{\it Acknowledgments.---}
%\section*{Acknowledgements}
The authors thank Jean-Michel Raimond, Serge Haroche, Benjamin Yadin, and Changhun Oh for helpful discussions.
This work was supported by a National Research Foundation of Korea grant funded by the Korea government (Grant No. 2010-0018295) and by the Korea Institute of Science and Technology Institutional Program (Project No. 2E27800-18-P043). K.C.T and T.J.V. were supported by Korea Research Fellowship Program through the National Research Foundation of Korea (NRF) funded by the Ministry of Science and ICT (Grants No. 2016H1D3A1938100 and 2016H1D3A1908876).
\end{acknowledgments}

\newpage
\widetext
\section{Supplemental Material}
\appendix
\section{Preliminaries}
For notational simplicity, we first define the optimal and mean values of quantum Fisher information (QFI) with respect to the multimode quadrature operator as follows:
\begin{equation}
\label{QF1}
{\cal I}_{\rm opt} (\hat\rho) := \frac{1}{2} \max_{\boldsymbol\mu \in {\cal S}} I_F(\hat\rho, \hat{X}_{\boldsymbol\mu}) = \frac{ \lambda_{\max} ({\boldsymbol F}) }{2}
\end{equation}
and
\begin{equation}
\label{QF2}
{\cal I}_{\rm mean}(\hat\rho) := \frac{1}{2} \int_{\cal S} \frac{d^{2N} \boldsymbol\mu}{{\rm Vol}({\cal S})} I_F(\hat\rho, \hat{X}_{\boldsymbol\mu}) = \frac{ {\rm Tr}{(\boldsymbol F})}{4N},
\end{equation}
where ${\cal S} = \{ \boldsymbol\mu | \sum_{n=1}^{N} |\mu_n|^2 = 1 \}$ and $\lambda_{\max} ({\boldsymbol F})$ is the maximum eigenvalue of $\boldsymbol F$. In the case of a pure state, the mean quadrature QFI coincides with the mean quadrature variance, i.e. $$\overline{\cal V} (\ket{\psi}) = \frac{1}{N} \sum_{k=1}^{2N} {\rm Var}(\psi, \hat{R}^{(k)}) = {\cal I}_{\rm mean}(\ket{\psi}),$$
where $\hat{R}^{(k)}$ is the $k$th element of $\hat{\boldsymbol R} := (\hat{x}_1, \hat{p}_1, \cdots, \hat{x}_N, \hat{p}_N)^T$ and $\hat{X}_{\boldsymbol \mu} = \hat{\boldsymbol R}^T {\boldsymbol \mu}$. For mixed states, the metrological power is given by $${\cal M}(\hat\rho) = {\cal I}_{\rm opt} (\hat\rho \otimes \ket{0}\bra{0}) -1 ={\rm max} \{{\cal I}_{\rm opt} (\hat\rho) -1, 0 \}.$$ Then we prove the following:

\begin{proposition} [Quadrature QFI of a continuous-variable system] Both ${\cal I}_{\rm opt}$ and ${\cal I}_{\rm mean}$ satisfy the following properties:
\label{LOF}
~
\begin{enumerate}
\item For a pure state $\ket{\psi}$, ${\cal I}(\ket{\psi}) \geq 1$, where equality holds {\it if and only if} $\ket{\psi}$ is a coherent state $\ket{\boldsymbol\alpha}$. For a mixed state $\hat\rho$, ${\cal I}(\hat\rho) \leq 1$ when $\hat\rho$ is classical. 

\item Linear optical unitaries do not change the quadrature QFI, i.e. 
${\cal I}(\hat{U}_{L} \hat\rho \hat{U}_{L}^{\dagger}) = {\cal I} (\hat\rho)$.

\item Combining two optical fields $A$ and $B$ cannot increase the overall quadrature QFI, i.e. ${\cal I} ( \hat{U}_L (\hat\rho_A \otimes \hat\sigma_B) \hat{U}_L^\dagger) \leq \max \{ {\cal I}(\hat\rho_A), {\cal I}(\hat\sigma_B) \}$, regardless of the number of modes in each field.
\end{enumerate}
\end{proposition}

\begin{proof}
We first show that ${\cal I}(\ket{\psi}) \geq 1$.
Taking a collective quadrature observable $\hat{X}_{\boldsymbol\mu}$ and its conjugate operator $\hat{P}_{\boldsymbol\mu} = (\hat{a}_{\boldsymbol \mu} - \hat{a}_{\boldsymbol \mu}^\dagger )/(\sqrt{2} i) = \hat{X}_{ \tilde{\boldsymbol\mu}}$, with $\tilde{\mu}_n = -i \mu_n$ for every $n$. Then, we note that $[\hat{X}_{\boldsymbol\mu}, \hat{P}_{\boldsymbol\mu}]= i$. Then we have
$$
\begin{aligned}
{\rm Var}(\psi, \hat{X}_{\boldsymbol\mu}) + {\rm Var}(\psi, \hat{P}_{\boldsymbol\mu}) 
\geq 2 \sqrt{{\rm Var}(\psi, \hat{X}_{\boldsymbol\mu}) {\rm Var}(\psi, \hat{P}_{\boldsymbol\mu})} 
\geq  | \langle [\hat{X}_{\boldsymbol\mu}, \hat{P}_{\boldsymbol\mu}] \rangle| = 1,
\end{aligned}
$$
where the second inequality is the Heisenberg uncertainty relation.
Using the fact that $I_F(\ket{\psi}, \hat{X}_{\boldsymbol\mu}) = 4 {\rm Var}(\psi, \hat{X}_{\boldsymbol\mu})$,
$$
\begin{aligned}
{\cal I}_{\rm opt}(\ket{\psi}) &= \frac{1}{2} \max_{\boldsymbol\mu \in {\cal S}} I_F(\ket{\psi}, \hat{X}_{\boldsymbol\mu}) \geq \frac{1}{4} \left[  I_F(\ket{\psi}, \hat{X}_{\boldsymbol\mu})  +  I_F(\ket{\psi}, \hat{P}_{\boldsymbol\mu}) \right] = {\rm Var}(\psi, \hat{X}_{\boldsymbol\mu}) + {\rm Var}(\psi, \hat{P}_{\boldsymbol\mu})  \geq 1,
\end{aligned}
$$
and
$$
\begin{aligned}
{\cal I}_{\rm mean} (\ket{\psi})  &= \frac{1}{2} \int_{\cal S} \frac{d^{2N} \boldsymbol\mu}{{\rm Vol}({\cal S})} I_F(\ket{\psi}, \hat{X}_{\boldsymbol\mu}) \\
& = \frac{1}{4} \left[ \int_{\cal S} \frac{d^{2N} \boldsymbol\mu}{{\rm Vol}({\cal S})} I_F(\ket{\psi}, \hat{X}_{\boldsymbol\mu}) + \int_{\cal S} \frac{d^{2N} \tilde{\boldsymbol\mu}}{{\rm Vol}({\cal S})} I_F(\ket{\psi}, \hat{X}_{\tilde{\boldsymbol\mu}}) \right] \\
&= \int_{\cal S} \frac{d^{2N} \boldsymbol\mu}{{\rm Vol}({\cal S})} \left[ {\rm Var}(\psi, \hat{X}_{\boldsymbol\mu}) +  {\rm Var}(\psi, \hat{P}_{\boldsymbol\mu})\right] \\
&\geq 1,
\end{aligned}
$$
since $d^{2N} \boldsymbol\mu = d^{2N} \tilde{\boldsymbol\mu}$ when integrating over ${\cal S}$.

We now show that the equality holds  {\it if and only if} $\ket{\psi}$ is a coherent state.
Note that ${\cal I}_{\rm opt} (\ket{\psi}) = 1$ or ${\cal I}_{\rm mean} (\ket{\psi}) = 1$ is equivalent to the condition $(1/2) I_F( \ket{\psi}, \hat{X}_{\boldsymbol\mu}) = 1$ for any $\boldsymbol\mu \in {\cal S}$.
Then the ``{\it if} " part can be verified by directly showing that for any coherent state $\ket{\boldsymbol\alpha} = \ket{\alpha_1} \ket{\alpha_2} \cdots \ket{\alpha_N}$.
$$
\begin{aligned}
\frac{1}{2}I_F(\ket{\boldsymbol\alpha}, \hat{X}_{\boldsymbol\mu}) &= \frac{1}{2}  I_F \left(\ket{\alpha_1} \ket{\alpha_2} \cdots \ket{\alpha_N}, \sum_{n=1}^N \frac{\mu_n^* \hat{a}_n + \mu_n \hat{a}_n^\dagger}{\sqrt{2}} \right)= \frac{1}{2}  \sum_{n=1}^N |\mu_n|^2 I_F(\ket{\alpha_n}, \hat{X}_{\boldsymbol \mu_n}) =1,
\end{aligned}
$$
where $\hat{X}_{\boldsymbol\mu_n} = |\mu_n|^{-1} (\mu_n^* \hat{a}_n + \mu_n \hat{a}_n^\dagger )/\sqrt{2}$ is a single-mode quadrature operator of the $n$-th mode. We used the fact that the QFI of any quadrature observable $\hat{x}_\theta = (\hat{a} e^{-i\theta} + \hat{a}^\dagger e^{i\theta})/\sqrt{2}$ is given by $I_F(\ket{\alpha}\bra{\alpha}, \hat{x}_{\theta}) = 2$ for any single-mode coherent state $\ket{\alpha}$.

The ``{\it only if} " part can be proved as follows.
The $N$-mode pure optical state $\ket{\psi}$ can be decomposed into one selected mode, say the first mode, and remaining modes, $\ket{\psi} = \sum_j \sqrt{\lambda_j} | \phi_j^{(1)} \rangle| \phi_j^{(2 \cdots N)} \rangle$, by the Schmidt decomposition.
We note that $\hat{R}^{(1)} = \hat{x}_1$ and $\hat{R}^{(2)} = \hat{p}_1$ then
$$
\begin{aligned}
{\rm Var}(\psi, \hat{R}^{(1)}) + {\rm Var}(\psi, \hat{R}^{(2)}) & = \sum_j \lambda_j \left[ \langle \phi_j^{(1)} | \hat{x}^2_1 |\phi_j^{(1)} \rangle +\langle \phi_j^{(1)} | \hat{p}^2_1 |\phi_j^{(1)} \rangle \right]  - \left[ \sum_j \lambda_j \langle \phi_j^{(1)} | \hat{x}_1 |\phi_j^{(1)} \rangle \right]^2 - \left[ \sum_j \lambda_j \langle \phi_j^{(1)} | \hat{p}_1 |\phi_j^{(1)} \rangle \right]^2 \\
&\quad \geq \sum_j \lambda_j \left[ {\rm Var}(\phi^{(1)}_j, \hat{x}_1) +  {\rm Var}(\phi^{(1)}_j, \hat{p}_1)  \right] \\
&\quad \geq 1,
\end{aligned}
$$
where $\hat{x}_1 = \frac{\hat{a}_1 + \hat{a}_1^\dagger}{\sqrt{2}}$ and $\hat{p}_1 = \frac{\hat{a}_1 - \hat{a}_1^\dagger}{\sqrt{2}i}$ are quadrature operators on the first mode.
The first inequality saturates only if every $\langle \phi_j^{(1)} | \hat{x}_1 |\phi_j^{(1)} \rangle$ gives the same value and $\langle \phi_j^{(1)} | \hat{p}_1 |\phi_j^{(1)} \rangle$ also gives the same value for all $j$ ($\langle \phi_j^{(1)} | \hat{x}_1 |\phi_j^{(1)} \rangle$  and $\langle \phi_j^{(1)} | \hat{p}_1 |\phi_j^{(1)} \rangle$ are not necessarily the same).
The second inequality is the Heisenberg uncertainty and only coherent state $| \phi_j^{(1)} \rangle = |\alpha_j^{(1)} \rangle$ reaches the bound 1.
Thus, combining these two results, $(1/2)I_F(\ket\psi, \hat{X}_{\boldsymbol\mu}) = 1$ for all $\boldsymbol\mu \in {\cal S}$ implies that 
${\rm Var}(\psi, \hat{R}^{(1)}) + {\rm Var}(\psi, \hat{R}^{(2)}) = 1$, thus $\ket{\psi}$ should be written in the form $\ket{\psi} = \ket{\alpha^{(1)}} \ket{\phi^{(2 \cdots N}}$, since a coherent state is uniquely defined by given $\langle \hat{x} \rangle$ and $\langle \hat{p} \rangle$.
We can repeat the same process for each mode, then we finally conclude that $\ket{\psi} = \ket{\alpha^{(1)}} \cdots \ket{\alpha^{(N)}}$ is a multimode coherent state.

Now we prove that ${\cal I}(\hat\rho) \leq 1 $ when $\hat\rho$ is classical.
Note that a classical state can be expressed as a convex sum of multimode coherent states, i.e. $\hat\rho = \sum_i p_i \ket{\boldsymbol\alpha_i} \bra{\boldsymbol\alpha_i}$ for positive $p_i$ satisfying $\sum_i p_i = 1$.
Then by the convexity of the QFI, we get $I_F(\hat\rho, \hat{X}_{\boldsymbol{\mu}}) \leq \sum_i p_i  I_F(\ket{\boldsymbol\alpha_i}\bra{\boldsymbol\alpha_i}, \hat{X}_{\boldsymbol\mu}) \leq 2$, where $I_F(\ket{\boldsymbol\alpha_i}\bra{\boldsymbol\alpha_i}, \hat{X}_{\boldsymbol\mu}) = 2$ for any coherent state $\ket{\boldsymbol\alpha}$.
Then by either taking optimization over $\boldsymbol\mu$ or averaging over  ${\cal S} = \{ \boldsymbol\mu | \sum_{n=1}^{N} |\mu_n|^2 = 1 \}$, we have ${\cal I}_{\rm opt} \leq 1$ and ${\cal I}_{\rm mean} \leq 1$, respectively.

Next, we demonstrate that the quadrature QFI ${\cal I}$ is invariant under optical linear unitaries.
Note that 
$I_F(\hat{U}_L \hat\rho \hat{U}_L^\dagger, \hat{X}_{\boldsymbol{\mu}}) = I_F(\hat\rho, \hat{U}_L^\dagger \hat{X}_{\boldsymbol{\mu}} \hat{U}_L) = I_F(\hat\rho, \hat{X}_{\boldsymbol{\mu'}})$, where there exist an $2N \times 2N$ unitary matrix $V$ such that $\boldsymbol\mu' = V \boldsymbol\mu \in {\cal S}$.
Moreover the unitary matrix $V$ does not change the structure of ${\cal S}$, i.e. $V {\cal S} V^\dagger = {\cal S}$ and $|\det V |= 1 $, which guarantees that the optimal and mean quadrature QFIs do not change by such unitaries.

Finally, we show that combining two-uncorrelated optical systems cannot increase the quadrature QFIs.
Suppose that the optical field $A$ and $B$ of $N$ and $M$ modes, respectively.
Then a collective bosonic operator in the combined system can be expressed using a $2(N +M)$ dimensional real vector $\boldsymbol\mu$.
Let us assume that the optimal quadrature QFI ${\cal I}_{\rm opt}$ of the combined system is given by 
$
{\cal I}_{\rm opt} (\hat{U}_L (\hat\rho \otimes \hat\sigma)  \hat{U}_L^\dagger) = {\cal I}_{\rm opt} (\hat\rho \otimes \hat\sigma) = (1/2)\max_{\boldsymbol\mu \in {\cal S}} I_F (\hat\rho \otimes \hat\sigma, \hat{X}_{\boldsymbol{\mu}}) = (1/2)I_F (\hat\rho \otimes \hat\sigma, \hat{X}_{\boldsymbol{\tilde{\mu}}}),
$
where the maximum is achieved at $\{ \tilde\mu_n \}$
% $\boldsymbol{\tilde{\mu}}=  (\tilde\mu_1, \tilde\mu_2, \cdots, \tilde\mu_{N+M})^T$
and $\hat{U}_L$ does not change the quadrature QFI $\cal P_{\rm opt}$.
Note that the quadrature $\hat{X}_{\boldsymbol{\tilde{\mu}}}$
can be divided into two parts $\hat{X}_{\boldsymbol{\tilde{\mu}}} = \sqrt\lambda \hat{X}_{\boldsymbol\mu_A} \otimes {\mathbb 1}_B + \sqrt{1-\lambda} {\mathbb 1}_A \otimes  \hat{X}_{\boldsymbol\mu_B}$, where $\boldsymbol{\mu}_A  = \lambda^{-1/2} ({\rm Re}\tilde{\mu}_1, {\rm Im}\tilde{\mu}_1, {\rm Re}, \cdots, {\rm Re}\tilde{\mu}_N, {\rm Im}\tilde{\mu}_N)^T$ and $\boldsymbol{\mu}_B  = (1-\lambda)^{-1/2} ({\rm Re}\tilde{\mu}_{N+1}, {\rm Im}\tilde{\mu}_{N+1}, \cdots, {\rm Re}\tilde{\mu}_{N+M}, {\rm Im}\tilde{\mu}_{N+M})^T$ with $\lambda = \sum_{i=1}^N |\tilde\mu_i|^2$.
Note that $|\boldsymbol\mu_A|^2 =  1 = |\boldsymbol\mu_B|^2$.
Then we have 
\begin{equation}
\label{EQ3}
\begin{aligned}
{\cal I}_{\rm opt}  (\hat\rho \otimes \hat\sigma) & = \frac{1}{2} I_F (\hat\rho \otimes \hat\sigma, \hat{X}_{\boldsymbol{\tilde\mu}}) \\
&= \frac{1}{2}   I_F (\hat\rho \otimes \hat\sigma,\sqrt\lambda \hat{X}_{\boldsymbol\mu_A} \otimes {\mathbb 1}_B + \sqrt{1-\lambda} {\mathbb 1}_A \otimes  \hat{X}_{\boldsymbol\mu_B})  \\
&= \frac{1}{2}  \left[ \lambda I_F (\hat\rho , \hat{X}_{\boldsymbol\mu_A})
+ (1-\lambda) I_F(\hat\sigma,  \hat{X}_{\boldsymbol\mu_B}) \right] \\
&\leq \max \{ {\cal I}_{\rm opt} (\hat\rho), {\cal I}_{\rm opt} (\hat\sigma) \},
\end{aligned}
\end{equation}
where we used the fact that $\frac{1}{2} I_F(\hat\rho, \hat{X}_{\boldsymbol{\mu_A}}) \leq {\cal I}_{\rm opt}(\rho)$ and  $\frac{1}{2} I_F(\hat\sigma, \hat{X}_{\boldsymbol{\mu_B}}) \leq {\cal I}_{\rm opt}(\sigma)$ and $0 \leq \lambda \leq 1$ for the last inequality.

For the mean quadrature QFI, we note that 
$$
{\cal I}_{\rm mean} (\hat\rho) = \frac{ {\rm Tr}{(\bold F})}{4N}= \frac{1}{4N} \sum_{k=1}^{2N} I_F(\hat\rho, \hat{R}^{(k)}).
$$
We have then the same property for the mean quadrature QFI,
$$
\begin{aligned}
{\cal I}_{\rm mean} (\hat\rho \otimes \hat\sigma) 
=  \frac{1}{4(N+M)} \sum_{k=1}^{2(N+M)} I_F(\hat\rho \otimes \hat\sigma, \hat{R}^{(k)}) 
= \frac{N {\cal I}_{\rm mean} (\hat\rho) + M {\cal I}_{\rm mean} (\hat\sigma)}{N+M} 
 \leq \max \{ {\cal I}_{\rm mean} (\hat\rho), {\cal I}_{\rm mean} (\hat\sigma)\}.
\end{aligned}
$$
\end{proof}

We also show the following properties of the QFI:
\begin{lemma} 
\label{FIP}
For a quantum classical state $\sum_i p_i \hat\rho^{(i)}_A \otimes \ket{i}_B\bra{i}$ with orthogonal basis $\{\ket{i}_B\}$, the QFI with respect to the given local observable $\hat{L}_A$ is given by
\begin{equation}
I_F \left(\sum_i p_i \hat\rho^{(i)}_A \otimes \ket{i}_B\bra{i} , \hat{L}_A \otimes {\mathbb 1}_B \right) = \sum_i p_i I_F(\hat\rho^{(i)}_A, \hat{L}_A).
\end{equation}
\end{lemma}
\begin{proof}
Note that eigenvalues of $\sum_i p_i \hat\rho^{(i)}_A \otimes \ket{i}_B\bra{i}$ are given by $p_i \lambda^{(i)}_\mu$, where $\lambda^{(i)}_\mu$ are eigenvalues of $\hat{\rho}^{(i)}_A$ with corresponding eigenstates $\ket{\phi_\mu^{(i)}}$.
By direct calculation, we have 
\begin{equation}
\begin{aligned}
I_F \left(\sum_i p_i \hat\rho^{(i)}_A \otimes \ket{i}_B\bra{i} , \hat{L}_A \otimes {\mathbb 1}_B \right) &= 2 \sum_{i,j} \sum_{\mu, \nu} \frac{\left(p_i \lambda^{(i)}_\mu - p_j \lambda^{(j)}_\nu \right)^2}{p_i \lambda^{(i)}_\mu + p_j \lambda^{(j)}_\nu} |\bra{\phi_\mu^{(i)}}_A \bra{i}_B \left( \hat{L}_A \otimes {\mathbb 1}_B \right) |\phi_\nu^{(j)} \rangle_A \ket{j}_B|^2 \\
&= 2 \sum_{i,j} \sum_{\mu, \nu} \frac{\left(p_i \lambda^{(i)}_\mu - p_j \lambda^{(j)}_\nu \right)^2}{p_i \lambda^{(i)}_\mu + p_j \lambda^{(j)}_\nu} |\bra{\phi_\mu^{(i)}} \hat{L}_A |\phi_\nu^{(j)} \rangle_A |^2 |\langle i | j\rangle_B|^2 \\
&= 2 \sum_{i} p_i \sum_{\mu, \nu} \frac{\left(\lambda^{(i)}_\mu - \lambda^{(i)}_\nu \right)^2}{\lambda^{(i)}_\mu + \lambda^{(i)}_\nu} \left| \langle\phi_\mu^{(i)} | \hat{L}_A |\phi_\nu^{(i)} \rangle_A \right|^2\\
&= \sum_i p_i I_F(\hat\rho^{(i)}_A, \hat{L}_A).
\end{aligned}
\end{equation}
\end{proof}

\begin{lemma} For an arbitrary quantum state $\hat\rho$ and Hermitian operators $\hat{A}$ and $\hat{B}$,
\label{QFIBound}
\begin{equation}
\begin{aligned}
\left| \sqrt{I_F(\hat\rho, \hat{A})} - \sqrt{I_F(\hat\rho, \hat{B})} \right|  \leq \sqrt{ I_F(\hat\rho, \hat{A} + \hat{B}) } \leq  \sqrt{I_F(\hat\rho, \hat{A})} + \sqrt{I_F(\hat\rho, \hat{B})}.
\end{aligned}
\end{equation}
\end{lemma}
\begin{proof}
Note that the QFI is always positive and can be rewritten as
$$
\begin{aligned}
I_F(\hat\rho, \hat{A} + \hat{B}) &= 2 \sum_{i,j} \frac{(\lambda_i - \lambda_j)^2}{\lambda_i + \lambda_j} | \bra{i} \hat{A} + \hat{B} \ket{j}|^2 \\
&= I_F(\hat\rho, \hat{A}) + I_F(\hat\rho, \hat{B}) + 4 \sum_{i,j} \frac{(\lambda_i - \lambda_j)^2}{\lambda_i + \lambda_j} \bra{i} \hat{A} \ket{j} \bra{j} \hat{B} \ket{i},
\end{aligned}
$$
due to its symmetry under exchanging $i$ and $j$. Then the last term, which has a real value, is bounded by
$$
\begin{aligned}
\left| \sum_{i,j} \frac{(\lambda_i - \lambda_j)^2}{\lambda_i + \lambda_j} \bra{i} \hat{A} \ket{j} \bra{j} \hat{B} \ket{i} \right| 
&\leq \sqrt{ \sum_{i,j} \frac{(\lambda_i - \lambda_j)^2}{\lambda_i + \lambda_j} |\bra{i} \hat{A} \ket{j}|^2} \sqrt{ \sum_{i,j} \frac{(\lambda_i - \lambda_j)^2}{\lambda_i + \lambda_j} |\bra{i} \hat{B} \ket{j}|^2} \\
&= \frac{1}{2} \sqrt{I_F(\hat\rho, \hat{A}) I_F(\hat\rho, \hat{B})},
\end{aligned}
$$
by using the Cauchy-Schwarz inequality. Finally, we have the claimed statement as
$$
\begin{aligned}
I_F(\hat\rho, \hat{A}) + I_F(\hat\rho, \hat{B}) - 2\sqrt{I_F(\hat\rho, \hat{A}) I_F(\hat\rho, \hat{B})}  \leq &I_F(\hat\rho, \hat{A} + \hat{B})  \leq I_F(\hat\rho, \hat{A}) + I_F(\hat\rho, \hat{B}) + 2\sqrt{I_F(\hat\rho, \hat{A}) I_F(\hat\rho, \hat{B})}\\
\Longrightarrow \left| \sqrt{I_F(\hat\rho, \hat{A})} - \sqrt{I_F(\hat\rho, \hat{B})} \right|  \leq &\sqrt{ I_F(\hat\rho, \hat{A} + \hat{B}) } \leq  \sqrt{I_F(\hat\rho, \hat{A})} + \sqrt{I_F(\hat\rho, \hat{B})}.
\end{aligned}
$$
\end{proof}

\section{Proof of Theorem~1}
\begin{proof}
Here we prove several important properties related to the nonclassicality measure 
$$
{\cal Q}(\hat\rho) := \min_{\{p_i, \ket{\psi_i} \}} \sum_i p_i \overline{\cal V}(\ket{\psi_i}) -1,
$$
 given by the convex roof of $\overline{\cal V}$.

{\it 1. Faithfulness condition:}
 If $\hat\rho$ is classical, the state can be represented as a convex sum of coherent states $\ket{\boldsymbol\alpha}$, which leads to ${\cal Q}(\hat\rho) = 0$ since $\overline{\cal V}(\ket{\boldsymbol\alpha}) = 1$.
Conversely, if ${\cal Q}(\hat\rho) = \min_{\{p_i, \psi_i \}} \sum_i p_i \overline{\cal V}(\ket{\psi_i}) -1= 0 $, there exists a pure state composition of $\hat\rho = \sum_i p_i^* \ket{\psi_i^*} \bra{\psi_i^*}$ that $\overline{\cal V}(\ket{\psi_i^*}) = 1 $ for every $i$. Note that $\overline{\cal V}(\ket{\psi_i^*}) = 1 $ {\it if and only if} $\ket{\psi_i^*}$ is a coherent state by Proposition~1, and using this decomposition, $\hat\rho$ can be expressed as a convex sum of coherent states, i.e. $\hat\rho$ is classical.
Thus we ${\cal Q}(\hat\rho) =0$ {\it if and only if} $\hat\rho$ is classical.\\

{\it 3. Convexity:} Convexity is guaranteed by the convex roof construction.\\

{\it 2-(a). Strong monotonicity condition:} In order to prove the strong monotonicity of the mean quadrature QFI, we first show that $\sum_i q_i \overline{\cal V}(\hat{K}_i \ket{\psi} / \sqrt{q_i}) \leq \overline{\cal V}(\ket{\psi})$ for a set of Kraus operators $\{ \hat{K}_i \}$ for a linear optical map with $q_i = \Tr \bra{\psi} \hat{K}_i^\dagger \hat{K}_i \ket{\psi}$. We let $\ket{\phi_i} = \hat{K}_i \ket{\psi} / \sqrt{q_i}$, then there exists a classical state $\hat\sigma_{EE'}$ and a linear optical unitary $\hat{U}_L$, such that ${\rm Tr}_E \hat{U}_L \left( \ket{\psi}\bra{\psi} \otimes \hat\sigma_{EE'} \right) \hat{U}_L^\dagger = \sum_i q_i \ket{\phi_i} \bra{\phi_i} \otimes \ket{i}_{E'}\bra{i}$.
Note that $\overline{\cal V} (\ket{\phi_i}) = (4N)^{-1} \sum_{k=1}^{2N} I_F(\ket{\phi_i}\bra{\phi_i}, \hat{R}^{(k)})$, then we get
\begin{equation*}
\begin{aligned}
\sum_i q_i \overline{\cal V} ( \ket{\phi_i})  & = \frac{1}{4N} \sum_i q_i  \sum_{k=1}^{2N} I_F \left(  \ket{\phi_i}\bra{\phi_i}, \hat{R}^{(k)} \right)   \\
& = \frac{1}{4N} \sum_{k=1}^{2N} I_F \left( \sum_i q_i \ket{\phi_i}\bra{\phi_i} \otimes \ket{i}_{E'}\bra{i}, \hat{R}^{(k)} \otimes {\mathbb 1}_{E'} \right)   \\
& = \frac{1}{4N} \sum_{k=1}^{2N} I_F \left( {\rm Tr}_E \hat{U}_L \left( \ket{\psi}\bra{\psi} \otimes \hat{\sigma}_{EE'}\right) \hat{U}_L^\dagger, \hat{R}^{(k)} \otimes {\mathbb 1}_{E'} \right) \\
& \leq \frac{1}{4N} \sum_{k=1}^{2N} I_F \left( \hat{U}_L \left( \ket{\psi}\bra{\psi} \otimes \hat{\sigma}_{EE'}\right) \hat{U}_L^\dagger, \hat{R}^{(k)} \otimes {\mathbb 1}_{EE'} \right) \\
& \leq \sum_j \frac{p_j}{4N} \sum_{k=1}^{2N} I_F \big( \hat{U}_L \left( \ket{\psi}\bra{\psi} \otimes \ket{\boldsymbol\alpha_j}_{EE'}\bra{\boldsymbol\alpha_j}\right) \hat{U}_L^\dagger,  \hat{R}^{(k)} \otimes {\mathbb 1}_{EE'} \big) ,
\end{aligned}
\end{equation*}
where $\hat{\sigma}_{EE'} = \sum_j p_j \ket{\boldsymbol\alpha_j} \bra{\boldsymbol\alpha_j}$ is classical, and the first and second inequalities come from contractivity and convexity of the QFI, respectively.
We prove that $\sum_{k=1}^{2N} I_F \left( \hat{U}_L \left( \ket{\psi}\bra{\psi} \otimes \ket{\boldsymbol\alpha_j}_{EE'}\bra{\boldsymbol\alpha_j}\right) \hat{U}_L^\dagger, \hat{R}^{(k)} \otimes {\mathbb 1}_{EE'} \right) \leq \sum_{k=1}^{2N} I_F \left( \ket{\psi}\bra{\psi}, \hat{R}^{(k)} \right)$.
Suppose $\hat{\sigma}_{EE'}$ is a $M$-mode classical state. We have
$$
\begin{aligned}
&\sum_{k=1}^{2N} I_F \left( \hat{U}_L \left( \ket{\psi}\bra{\psi} \otimes \ket{\boldsymbol\alpha_j}_{EE'}\bra{\boldsymbol\alpha_j}\right) \hat{U}_L^\dagger, \hat{R}^{(k)} \otimes {\mathbb 1}_{EE'} \right) \\
&= \sum_{k=1}^{2(N+M)} I_F \left( \hat{U}_L \left( \ket{\psi}\bra{\psi} \otimes \ket{\boldsymbol\alpha_j}_{EE'}\bra{\boldsymbol\alpha_j}\right) \hat{U}_L^\dagger, \hat{R}^{(k)} \right)  - \sum_{l=1}^{2M} I_F \left( \hat{U}_L \left( \ket{\psi}\bra{\psi} \otimes \ket{\boldsymbol\alpha_j}_{EE'}\bra{\boldsymbol\alpha_j}\right) \hat{U}_L^\dagger, {\mathbb 1} \otimes \hat{R}^{(l)}_{EE'} \right)
\\
&\leq \sum_{k=1}^{2(N+M)} I_F \left( \ket{\psi}\bra{\psi} \otimes \ket{\boldsymbol\alpha_j}_{EE'}\bra{\boldsymbol\alpha_j}, \hat{R}^{(k)} \right)- 4M \\
&= \sum_{l=1}^{2N} I_F \left( \ket{\psi}\bra{\psi}, \hat{R}^{(l)} \right) + \sum_{l=1}^{2M} I_F \left(  \ket{\boldsymbol\alpha_j}_{EE'}\bra{\boldsymbol\alpha_j}, \hat{R}^{(l)} \right)  - 4M \\
&\leq \sum_{k=1}^{2N} I_F \left( \ket{\psi}\bra{\psi}, \hat{R}^{(k)} \right).
\end{aligned}
$$
Finally, we get 
$$
\begin{aligned}
\sum_i q_i \overline{\cal V} ( \ket{\phi_i}) &\leq \sum_j \frac{p_j}{4N} \sum_{k=1}^{2N} I_F \big( \hat{U}_L \left( \ket{\psi}\bra{\psi} \otimes \ket{\boldsymbol\alpha_j}_{EE'}\bra{\boldsymbol\alpha_j}\right) \hat{U}_L^\dagger,   \hat{R}^{(k)} \otimes {\mathbb 1}_{EE'} \big) \\
& \leq (4N)^{-1} \sum_i q_i \sum_{k=1}^{2N} I_F \left( \ket{\psi}\bra{\psi}, \hat{R}^{(k)} \right)   \\
&= \overline{\cal V}(\ket{\psi}).
\end{aligned}
$$

Now we prove the strong monotonicity of ${\cal Q}$. Suppose ${\cal Q}(\hat\rho) = \min_{\{p_\mu, \psi_\mu \}} \sum_\mu p_\mu \overline{\cal V}(\ket{\psi_\mu}) -1 = \sum_\mu p^*_\mu \overline{\cal V}(\ket{\psi^*_\mu}) -1 $, where the minimum is achieved at $\hat\rho = \sum_\mu p^*_\mu \ket{\psi^*_\mu} \bra{\psi^*_\mu}$. Then we have 
\begin{equation}
\begin{aligned}
\sum_i q_i {\cal Q}(\hat{K}_i \hat\rho \hat{K}_i^\dagger / q_i) &\leq \sum_i \sum_\mu p^*_\mu q^\mu_i \overline{\cal V} \left(\hat{K}_i \ket{\psi^*_\mu} / \sqrt{q^\mu_i} \right) -1 \leq \sum_\mu p^*_\mu \overline{\cal V}(\ket{\psi^*_\mu}) -1 ={\cal Q}(\hat\rho),
\end{aligned}
\end{equation}
where $\hat{K}_i \hat\rho \hat{K}_i^\dagger / q_i = \sum_\mu (p^*_\mu q^\mu_i /q_i) \left( \frac{\hat{K}_i \ket{\psi^*_\mu}}{\sqrt{q^\mu_i}}\right) \left( \frac{\bra{\psi^*_\mu} \hat{K}_i^\dagger}{\sqrt{q^\mu_i}} \right)$ with $q^\mu_i = \bra{\psi^*_\mu} \hat{K}_i^\dagger \hat{K}_i \ket{\psi^*_\mu}$.\\

{\it 2-(b). Weak monotonicity:}
Finally, weak monotonicity can be derived by the strong monotonicity condition and convexity of ${\cal Q}$:
\begin{equation}
{\cal Q}\left(\sum_i \hat{K}_i \hat\rho \hat{K}_i^\dagger \right) \leq \sum_i p_i {\cal Q} \left( \hat{K}_i \hat\rho \hat{K}_i^\dagger / p_i \right) \leq {\cal Q}\left(\hat\rho\right),
\end{equation}
where $p_i = \Tr \hat\rho \hat{K}_i^\dagger \hat{K}_i$.
\end{proof}

\section{Proof of Theorem~2}

\begin{proof} Here we prove several important properties related to the metrological power ${\cal M}$ given by ${\cal M}(\hat\rho) = {\cal I}_{\rm opt} (\hat\rho \otimes \ket{0}\bra{0}) -1$.

{\it 1. ${\cal M}(\hat\rho) = 0 $ if $\hat\rho$ is classical:}
A classical quantum state has a positive-$P$ representation, which allows us to represent $\hat\rho = \pi^{-N} \int d^{2N} \boldsymbol\alpha P_{\hat\rho}( \boldsymbol \alpha) \ket{\boldsymbol \alpha} \bra{ \boldsymbol \alpha}$ as a convex sum of coherent states.
Then by convexity of ${\cal M}$ and using the fact that ${\cal M}(\ket{\alpha} \bra{\alpha}) =0$ for any coherent states, we have ${\cal M}(\hat\rho) =0$.\\

{\it 2. Monotonicity under a linear optical map:}
We note that the optimal quadrature QFI is contractive under partial trace:
$$
\begin{aligned}
{\cal I}_{\rm opt}({\rm Tr}_B \hat\rho_{AB}) &= \frac{1}{2} \max_{\mu_A \in {\cal S_A}} I_F({\rm Tr}_B \hat\rho_{AB}, \hat{X}_{\mu_A}) \\
&\leq \frac{1}{2} I_F(\hat\rho_{AB}, \hat{X}_{\tilde\mu_A} \otimes {\mathbb 1}_B) \\
&\leq \frac{1}{2} \max_{\mu_{AB} \in {\cal S_{AB}}} I_F(\hat\rho_{AB}, \hat{X}_{\mu_{AB}}) \\
&= {\cal I}_{\rm opt} (\hat\rho_{AB}),
\end{aligned}
$$
where $\tilde\mu_A$ gives the maximum metrolgical power for ${\rm Tr}_B \hat\rho_{AB}$.
Note that a linear optical map on $\hat\rho$ can be realized by $\Phi_L(\hat\rho) = {\rm Tr}_E \hat{U}_L \left( \hat\rho \otimes \hat\sigma_E \right)\hat{U}_L^\dagger$, where $\hat\sigma_E$ is classical and $\hat{U}_L$ is given by combinations of beam splitter operations, phase rotations, and displacement operations.
Then by Proposition~1, we have 
$$
\begin{aligned}
{\cal I}_{\rm opt}(\Phi_L(\hat\rho)) &= {\cal I}_{\rm opt}({\rm Tr}_E \hat{U}_L (\hat\rho \otimes \hat\sigma_E) \hat{U}_L^\dagger ) \\& \leq {\cal I}_{\rm opt}( \hat{U}_L (\hat\rho \otimes \hat\sigma_E) \hat{U}_L^\dagger ) \\& \leq \max \{ {\cal I}_{\rm opt}(\hat\rho), {\cal I}(\hat\sigma_E) \} \\& \leq \max \{ {\cal I}_{\rm opt}(\hat\rho), 1 \},
\end{aligned}
$$
since ${\cal I}_{\rm opt}(\hat\sigma_E) \leq 1$.
Then we get 
$$
{\cal M}(\Phi_L(\hat\rho)) = \max \{ {\cal I}_{\rm opt}(\Phi_L(\hat\rho)) - 1, 0 \} \leq \max \left\{ {\cal I}_{\rm opt} (\hat\rho) - 1 , 0 \right\} = {\cal M}(\hat\rho).\\
$$

{\it 3. Convexity:} Convexity of ${\cal M}$ is guaranteed by convexity of the QFI.\\

{\it 4. ${\cal M}(\hat\rho_A \otimes \hat\sigma_B) = {\rm max} \{ {\cal M}(\hat\rho_A), {\cal M}(\hat\sigma_B) \}$:}
We now show the final condition by noting that  ${\cal I}_{\rm opt}(\hat\rho_A \otimes \hat\sigma_B) \leq {\rm max} \{ {\cal I}_{\rm opt}(\hat\rho_A), {\cal I}_{\rm opt}(\hat\sigma_B) \}$ from Proposition~\ref{LOF}.
Then we can always choose $\hat{X}_{\boldsymbol\mu}$ with $\boldsymbol\mu = ({\tilde{\boldsymbol\mu}_A}^T, \boldsymbol{ 0}, \cdots, \boldsymbol{0})^T$ or $\boldsymbol\mu = (\boldsymbol{ 0}, \cdots, \boldsymbol{0}, {\tilde{\boldsymbol\mu}_A}^T)^T$ to achieve ${\cal I}_{\rm opt} (\hat\rho_A \otimes \hat\sigma_B) = {\rm max} \{ {\cal I}_{\rm opt}(\hat\rho_A ), {\cal I}_{\rm opt}(\hat\sigma_B) \}$ (see Eq.~(\ref{EQ3})), where $\tilde{\boldsymbol\mu}_A$ and $\tilde{\boldsymbol\mu}_B$ give the maximum quadrature QFI for $\hat\rho$ and $\hat\sigma$, respectively. This condition also leads to the fact that 
$${\cal M}(\hat\rho) = {\cal I}_{\rm opt} (\hat\rho \otimes \ket{0}\bra{0}) -1 ={\rm max} \{{\cal I}_{\rm opt} (\hat\rho), {\cal I}_{\rm opt}(\ket{0}\bra{0})\} -1 = {\rm max} \{{\cal I}_{\rm opt}(\hat\rho) -1 ,0 \},$$
since ${\cal I}_{\rm opt}(\ket{0}\bra{0}) = 1$.
\end{proof}

\section{Metrological power of quantum states}
\subsection{Decohered cat states}
A decohered optical cat state is given as follows:
$$\hat\rho_{\Gamma} = \frac{1}{N_\Gamma} \left[ \ket{\alpha} \bra{\alpha} + \ket{-\alpha} \bra{-\alpha} + \Gamma (\ket{\alpha} \bra{-\alpha} + \ket{-\alpha} \bra{\alpha}) \right], $$
where positive and negative values of $\Gamma$ refer to decohered even and odd cat states, respectively and $N_\Gamma = 2 + 2 \Gamma e^{-2\alpha^2}$.
Also, $|\Gamma| \leq 1$ and a lower value of $|\Gamma|$ refers to a more decohered cat state.
Note that we can express a decohered cat state in terms of pure even $\ket{e} = (\ket{\alpha} + \ket{-\alpha})/\sqrt{N_e}$ and odd $\ket{o} = (\ket{\alpha} - \ket{-\alpha})/\sqrt{N_o}$ cat states, where $N_e = 2 + 2 e^{-2|\alpha|^2}$ and $N_o = 2 - 2 e^{-2|\alpha|^2} $.
In this orthogonal basis $\{ \ket{o}, \ket{e} \}$, we have
$$
\hat\rho_{\Gamma} = \frac{1}{2N_\Gamma} \left[ N_e (1 + \Gamma) \ket{e}\bra{e} + N_e(1 - \Gamma) \ket{o}\bra{o} \right].
$$
Then we have
$$
\begin{aligned}
I_F(\hat\rho_\Gamma, \hat{X}_\theta) &= 2 \sum_{i,j} \frac{(\lambda_i - \lambda_j)^2}{\lambda_i + \lambda_j} | \bra{i} \hat{X}_\theta \ket{j}|^2 \\
&= 4 {\rm Tr} [\hat\rho_\Gamma \hat{X}_\theta^2] - \sum_{i,j} \frac{8\lambda_i \lambda_j}{\lambda_i + \lambda_j} | \bra{i} \hat{X}_\theta \ket{j}|^2 \\
&= 4 {\rm Tr} [\hat\rho_\Gamma \hat{X}_\theta^2] - 16 \lambda_e \lambda_o | \bra{e} \hat{X}_\theta \ket{o}|^2  - 4 \left[ \lambda_e |\bra{e}\hat{X}_\theta \ket{e}|^2 + \lambda_o |\bra{o} \hat{X}_\theta \ket{o}|^2\right],
\end{aligned}
$$
where $\hat{X}_\theta = \frac{1}{\sqrt{2}} (\hat{a} e^{-i \theta} + \hat{a}^\dagger e^{i \theta})$, $\lambda_e = \frac{N_e (1 + \Gamma)}{2N_\Gamma} $, and $\lambda_o = \frac{N_o (1 - \Gamma)}{2N_\Gamma} $.
Direct calculation leads to
$$
\begin{aligned}
\bra{e}\hat{X}_\theta \ket{e} &= 0 = \bra{o}\hat{X}_\theta \ket{o} \\
\bra{e}\hat{X}_\theta^2 \ket{e} &= \frac{1}{2} \left[ 2\left( \frac{N_o}{N_e} \right)|\alpha|^2+ 1 + \alpha^2 e^{-2 i \theta} + (\alpha^*)^2 e^{2 i \theta} \right] = |\alpha|^2 \left[ \frac{N_o}{N_e} + \cos (2(\theta - \phi))\right] + \frac{1}{2}  \\
\bra{o}\hat{X}_\theta^2 \ket{o}  &=\frac{1}{2} \left[ 2\left( \frac{N_e}{N_o} \right)|\alpha|^2+ 1 + \alpha^2 e^{-2 i \theta} + (\alpha^*)^2 e^{2 i \theta} \right] = |\alpha|^2 \left[ \frac{N_e}{N_o} + \cos (2(\theta - \phi))\right] + \frac{1}{2}  \\
|\bra{e}\hat{X}_\theta \ket{o}|^2 &= \frac{1}{2} \left| \alpha e^{-i\theta} \sqrt{\frac{N_o}{N_e}} + \alpha^* e^{i\theta} \sqrt{\frac{N_e}{N_o}}  \right|^2 = \frac{|\alpha|^2}{2} \left[ \frac{N_e}{N_o} + \frac{N_o}{N_e} + 2 \cos (2(\theta - \phi)) \right],
\end{aligned}
$$
where $\alpha = |\alpha|e^{i\phi}$.
We then have 
$$
\begin{aligned}
I_F(\hat\rho_\Gamma, \hat{X}_\theta) =  \frac{16 |\alpha|^2}{N_\Gamma^2} \left[ (\Gamma + e^{-2|\alpha|^2})^2 \cos (2(\theta - \phi)) + (\Gamma^2 - e^{-4|\alpha|^2}) \right] + 2
\end{aligned}
$$
and obviously, $\theta = \phi$ gives the maximum value of quantum Fisher Information
$$
\max_\theta I_F(\hat\rho_\Gamma, \hat{X}_\theta) = 2 + \frac{32 |\alpha|^2}{N_\Gamma^2} \Gamma (\Gamma + e^{-2|\alpha|^2}).
$$
Thus, the metrological power of decohered cat states is given by 
$$
{\cal M}(\hat\rho_\Gamma) = \max\left\{ \frac{16 |\alpha|^2}{N_\Gamma^2} \Gamma (\Gamma + e^{-2|\alpha|^2}), 0 \right\}.
$$
Note that decohered cat states are nonclassical if and only if $\Gamma \neq 0$.
For $0< \Gamma \leq 1$, we can see that ${\cal M}(\hat\rho_\Gamma) >0$, which implies that every decohered even cat state gives nonvanishing value of ${\cal M}$ if and only if it is nonclassical state.
On the other hand, a decohered odd cat state with negative values of $\Gamma$ can give ${\cal M} =0$ when $-e^{-2 |\alpha|^2} \leq \Gamma <0$. This implies that some of odd cat states do not give quantum enhancement for displacement estimation tasks even they have negative $P$ representations.

\subsection{Gaussian states}
Suppose a generic $N$-mode Gaussain state $\hat{\rho}_{(\boldsymbol V, \boldsymbol d)}$ has an $2N \times 2N$ covariance matrix $\boldsymbol V$, and its the sympletic decomposition is given by $\boldsymbol V = \boldsymbol S \boldsymbol V^\oplus \boldsymbol S^T$, where $\boldsymbol V^\oplus = \bigoplus_{n=1}^N \nu_n \mathbb 1_n$ with $\boldsymbol S \boldsymbol \Omega \boldsymbol S^T = \boldsymbol \Omega$.
Using the fact that every sympletic matrix $\boldsymbol S$ corresponds to a multimode unitary operation $\hat{U}_{\boldsymbol S}$ (not necessarily be a linear optical map), we can express the state as follows:
$$
\hat{\rho}_{(\boldsymbol V,\boldsymbol d)} = \hat{U}_{\boldsymbol S} \left[ \bigotimes_{n=1}^N \hat{\tau}_n \right] \hat{U}_{\boldsymbol S}^\dagger,
$$
where $\hat{\tau}_n$ is a thermal state of $n$-th mode with mean-photon number $(\nu_n -1 )/2$.
Then the QFI with respect to $\hat{X}_{\boldsymbol \mu}$ is given by
$$
\begin{aligned}
I_F(\hat\rho_{(\boldsymbol V, \boldsymbol d)}, \hat{X}_{\boldsymbol \mu}) &= I_F( \hat{U}_{\boldsymbol S} \left( \otimes_{n=1}^N \hat{\tau}_n \right) \hat{U}_{\boldsymbol S}^\dagger, \hat{X}_{\boldsymbol \mu}) = I_F( \otimes_{n=1}^N \hat{\tau}_n, \hat{U}_{\boldsymbol S}^\dagger \hat{X}_{\boldsymbol \mu} \hat{U}_{\boldsymbol S}).
\end{aligned}
$$
It is important to note that $\hat{U}_{\boldsymbol S}^\dagger \hat{X}_{\boldsymbol \mu} \hat{U}_{\boldsymbol S} = \hat{X}_{\tilde{\boldsymbol \mu}}$, where
$\tilde{\boldsymbol \mu} = (\boldsymbol S^{-1})^T \boldsymbol \mu$, where $||\tilde{\boldsymbol \mu} ||^2 \neq 1 $ in general. 

Moreover, the QFI matrix for a multimode thermal state is given by $\boldsymbol F(\otimes_{n=1}^N \hat{\tau}_n) = 2 \oplus_{n=1}^N \nu_n^{-1} {\mathbb 1} _n = 2 (\boldsymbol V^{\oplus})^{-1}= 2 \boldsymbol S^T \boldsymbol V^{-1} \boldsymbol S$.
Using the expression of Eq.~(1), we get
$$
\begin{aligned}
I_F(\hat\rho_{(\boldsymbol V,\boldsymbol d)}, \hat{X}_{\boldsymbol \mu}) &= 2 ((\boldsymbol S^{-1})^T \boldsymbol \mu)^T \boldsymbol F(\otimes_{n=1}^N \hat{\tau}_n)  (\boldsymbol S^{-1})^T \boldsymbol \mu = 2 \boldsymbol \mu^T (\boldsymbol S^{-1} \boldsymbol S^T \boldsymbol V^{-1} \boldsymbol S (\boldsymbol S^{-1})^T) \boldsymbol \mu.
\end{aligned}
$$
We then obtain the desired formula:
$$
{\cal M}(\hat\rho_{(\boldsymbol {V,d})}) = \max \left\{ \lambda_{\max} \left[\boldsymbol S^{-1} \boldsymbol S^T \boldsymbol V^{-1} \boldsymbol S (\boldsymbol S^{-1})^T \right] - 1 , 0 \right\}.
$$
For a single-mode Gaussian state direct calculation leads to
$$
\boldsymbol S^{-1} \boldsymbol S^T \boldsymbol V^{-1} \boldsymbol S (\boldsymbol S^{-1})^T =
\left(
\begin{matrix}
\frac{e^{-2|\xi|}}{2\bar{n}_{\rm th}+1} & 0 \\
0 & \frac{e^{2|\xi|}}{2\bar{n}_{\rm th}+1}
\end{matrix}
\right),
$$
then ${\cal M} = \max \left\{ \exp(2 r) / (2\bar{n}_{\rm th} +1) -1, 0 \right\}$. Note that a single-mode Gaussian state is classical (i.e., has a positive $P$ representation) when $|\xi| \leq r_c = (1/2) \log (2 \bar{n}_{\rm th} +1)$. This proves {\bf Corollary 2.1} that {\it the metrological power ${\cal M}(\hat\rho_G) $ is zero if and only if a single-mode Gaussian state $\hat\rho_G$ is classical.}

\section{Proof of Theorems 3 and 4}
Analogous to the quadrature QFI, we define the $\alpha$-invested QFI for phase estimation as 
$$
{\cal I}_{\rm phase}^\alpha (\hat\rho) = \frac{1}{4} \max_{\hat{U}_L^\alpha} I_F( \hat{U}_L^\alpha \hat\rho \hat{U}_L^{\alpha \dagger}, \hat{a}^\dagger_1 \hat{a}_1)
$$
%then we have ${\cal M}_{\rm phase}^\alpha (\hat\rho)= {\cal I}_{\rm phase}^\alpha(\hat\rho \otimes \ket{0}\bra{0}) - {\rm Tr} [ \hat{U}_L^\alpha (\hat\rho \otimes \ket{0}\bra{0}) \hat{U}_L^{\alpha \dagger} \hat{a}_1^\dagger \hat{a}_1]$.
and prove the following Proposition:
\begin{proposition} [Relationship between the QFIs for phase and displacement estimation] The $\alpha$-invested QFI for phase estimation is bounded by the optimal quadrature QFI via the following:
\label{Prop2}
$$
\left[ \sqrt{{\cal I}_{\rm phase}^0(\hat\rho)} - |\alpha| \sqrt{{\cal I}_{\rm opt}(\hat\rho)}\right]^2 \leq {\cal I}_{\rm phase}^\alpha(\hat\rho) \leq \left[ \sqrt{{\cal I}_{\rm phase}^0(\hat\rho)} + |\alpha| \sqrt{{\cal I}_{\rm opt}(\hat\rho)}\right]^2,
$$
where ${\cal I}_{\rm phase}^0$ is the quadrature QFI for phase estimation without  any invested displacement operation.
\end{proposition}

\begin{proof}
We first prove the upper bound. Assume that ${\cal I}_{\rm phase}^\alpha$ reaches its maximal value by choosing a linear optical unitary $\hat{U}_L^{\alpha \star} = \left[ \prod_{n=1}^N \hat{D}_n(\alpha_n^\star) \right] \hat{V}_L^{0 \star}$, where $\alpha_n^\star = |\alpha_n^\star| e^{i\phi_n} $ and $\hat{V}_L^0$ is a passive linear optical unitary corresponding to an element of $O(2N)$, i.e., a rotation of canonical observables. Then we have
$$
\begin{aligned}
{\cal I}_{\rm phase}^\alpha(\hat\rho) 
&= \frac{1}{4} I_F(\hat{U}_L^{\alpha \star} \hat\rho \hat{U}_L^{\alpha \star \dagger}, \hat{a}_1^\dagger \hat{a}_1) \\
& = \frac{1}{4} I_F(\hat{V}_L^{0 \star} \hat\rho \hat{V}_L^{0 \star \dagger}, (\hat{a}_1^\dagger + \alpha_1^{\star *})(\hat{a}_1 + \alpha_1^{\star})  ) \\
& = \frac{1}{4} I_F\left(\hat{V}_L^{0 \star} \hat\rho \hat{V}_L^{0 \star \dagger}, \hat{a}_1^\dagger \hat{a}_1 + |\alpha_1^{\star}| ( \hat{a}_1 e^{-i\phi_1} + \hat{a}_1^\dagger e^{ i \phi_1} )\right)\\
& \leq \frac{1}{4} \Big[ \sqrt{I_F( \hat{V}_L^{0 \star} \hat\rho \hat{V}_L^{0 \star \dagger}, \hat{a}_1^\dagger \hat{a}_1)} + \sqrt{I_F( \hat{V}_L^{0 \star} \hat\rho \hat{V}_L^{0 \star \dagger},  |\alpha_1^{\star}| ( \hat{a}_1 e^{-i\phi_1} + \hat{a}_1^\dagger e^{ i \phi_1} )} \Big]^2 \\
&= \frac{1}{4} \Big[ \sqrt{I_F( \hat{V}_L^{0 \star} \hat\rho \hat{V}_L^{0 \star \dagger}, \hat{a}_1^\dagger \hat{a}_1)} + \sqrt{2} |\alpha_1^{\star}| \sqrt{I_F( \hat\rho, \hat{X}_{\boldsymbol \mu}) } \Big]^2 \\
& \leq \Big[ \sqrt{{\cal I}_{\rm phase}^0(\hat\rho)} + |\alpha|\sqrt{ {\cal I}_{\rm opt}(\hat\rho)} \Big]^2,
\end{aligned}
$$
where the first inequality is given by Lemma~\ref{QFIBound} and the second inequality is from optimality of ${\cal I}_{\rm phase}^0$ and ${\cal I}_{\rm opt}$ and $|\alpha| = \sqrt{\sum_{n=1}^N |\alpha_n^\star|^2 }\geq |\alpha_1^\star|$.

The lower bound can be proven by considering two different cases.
If ${\cal I}_{\rm phase}^0 (\hat\rho) \geq |\alpha|^2 {\cal I}_{\rm opt} (\hat\rho)$, we take $\hat{U}_L^\alpha = \hat{D}_1(|\alpha|) \hat{V}_L^{0}$, where $I_F (\hat{V}_L^{0} \hat\rho \hat{V}_L^{0 \dagger}, \hat{a}_1^\dagger \hat{a}_1) = 4{\cal I}_{\rm phase}^0(\hat\rho)$. If  ${\cal I}_{\rm phase}^0 (\hat\rho) < |\alpha|^2 {\cal I}_{\rm opt} (\hat\rho)$, we take $\hat{U}_L^\alpha = \hat{D}_1(|\alpha|) \hat{V}_L^{0}$, where $I_F (\hat{V}_L^{0} \hat\rho \hat{V}_L^{0 \dagger}, \hat{x}_1) = 2{\cal I}_{\rm opt}(\hat\rho)$. In either case, we have
$$
\begin{aligned}
{\cal I}_{\rm phase}^\alpha (\hat\rho) &\geq 
\frac{1}{4} I_F(\hat{U}_L^{\alpha} \hat\rho \hat{U}_L^{\alpha \dagger}, \hat{a}_1^\dagger \hat{a}_1) \\
&\geq \frac{1}{4} \Big[ \sqrt{I_F( \hat{V}_L^{0} \hat\rho \hat{V}_L^{0 \dagger}, \hat{a}_1^\dagger \hat{a}_1)} - \sqrt{2} |\alpha| \sqrt{I_F( \hat{V}_L^0 \hat\rho \hat{V}_L^{0 \dagger}, \hat{x}_1) } \Big]^2 \\
&\geq \Big[ \sqrt{{\cal I}_{\rm phase}^0(\hat\rho)} - |\alpha|\sqrt{ {\cal I}_{\rm opt}(\hat\rho)} \Big]^2,
\end{aligned}
$$
where the first (third) inequality comes from the optimality of ${\cal I}_{\rm phase}^\alpha$ (${\cal I}_{\rm phase}^0$ and ${\cal I}_{\rm opt}$) the second inequality is given by Lemma~\ref{QFIBound}.
\end{proof}

\subsection{Proof of Theorem~3}
\begin{proof}
The theorem can be proved by showing that there exists $\alpha$ and $\hat{U}_L^\alpha$ such that
\begin{equation}
I_F (\hat{U}_L^\alpha (\hat\rho \otimes \ket{0}\bra{0}) \hat{U}_L^{\alpha \dagger}, \hat{a}_1^\dagger \hat{a}_1) > 4 {\rm Tr} \left[ \hat{U}_L^\alpha  (\hat\rho \otimes \ket{0}\bra{0}) \hat{U}_L^{\alpha \dagger} \hat{a}_1^\dagger \hat{a}_1\right].\label{eqn:sss}
\end{equation} 
Because of the lower bound appearing in Proposition~\ref{Prop2}, for any $\alpha$, we can select a $\hat{U}_L^\alpha $ such that
\begin{eqnarray}
\frac{1}{4} I_F (\hat{U}_L^\alpha  (\hat\rho \otimes \ket{0}\bra{0}) \hat{U}_L^{\alpha \dagger}, \hat{a}_1^\dagger \hat{a}_1) &\geq & \left[ \sqrt{{\cal I}_{\rm phase}^0 (\hat\rho \otimes \ket{0}\bra{0})} - |\alpha| \sqrt{ {\cal I}_{\rm opt}  (\hat\rho \otimes \ket{0}\bra{0})} \right]^2 \nonumber \\
&=& |\alpha|^2 {\cal I}_{\rm opt}  (\hat\rho \otimes \ket{0}\bra{0}) + {\cal I}_{\rm phase}^0 (\hat\rho \otimes \ket{0}\bra{0})\nonumber \\ &{}&- 2 |\alpha| \sqrt{{\cal I}_{\rm phase}^0 (\hat\rho \otimes \ket{0}\bra{0}) {\cal I}_{\rm opt} (\hat\rho \otimes \ket{0}\bra{0})} .
\label{eqn:rrr}
\end{eqnarray}

With $\hat{U}_L^\alpha $ chosen in this way, one computes the right hand side of Eq.(\ref{eqn:sss}) by using $\hat{U}_L^\alpha =\left( \prod_{j=1}^{N} \hat{D}_{j}(\alpha_{j}) \right) \hat{V}_{L}^{0}$, where $\hat{V}_{L}^{0}$ is a number-conserving unitary:
$$
\begin{aligned}
{\rm Tr}\left[ \hat{U}_L^\alpha  (\hat\rho \otimes \ket{0}\bra{0}) \hat{U}_L^{\alpha \dagger} \hat{a}_1^\dagger \hat{a}_1 \right] &= {\rm Tr} \left[ \hat{V}_L^0 (\hat\rho \otimes \ket{0}\bra{0}) \hat{V}_L^{0 \dagger} \hat{a}_1^\dagger \hat{a}_1 \right] + |\alpha_{1}| {\rm Tr} \left[ \hat{V}_L^0 (\hat\rho \otimes \ket{0}\bra{0}) \hat{V}_L^{0 \dagger} (\hat{a}_1 e^{-i \phi_1} + \hat{a}_1^\dagger e^{i \phi_1} )\right] + |\alpha_{1}|^2 \\
&\leq \bar{n} + \sqrt{2}|\alpha | {\rm Tr} \left[ (\hat\rho \otimes \ket{0}\bra{0}) \hat{X}_{\boldsymbol \mu}\right] + |\alpha|^2 \\ &\leq \bar{n} + 2|\alpha | \sqrt{ \overline{n}+{N+1\over 2}} + |\alpha|^2
\end{aligned}
$$
where in the first line we have used $\phi_{1}=\text{Arg}\alpha_{1}$, and in the second line we have noted that ${\rm Tr} \left[ \hat{V}_L^0 (\hat\rho \otimes \ket{0}\bra{0}) \hat{V}_L^{0 \dagger} \hat{a}_1^\dagger \hat{a}_1 \right] \le {\rm Tr} \left[ \hat{V}_L^0 (\hat\rho \otimes \ket{0}\bra{0}) \hat{V}_L^{0 \dagger} \sum_{j=1}^{N}\hat{a}_j^\dagger \hat{a}_j \right] = \overline{n}$, $|\alpha_{1}|^{2}\le |\alpha |^{2}$, and defined ${\boldsymbol \mu} \in \mathbb{R}^{2N+2}$, $\Vert {\boldsymbol \mu} \Vert = 1$ such that $\hat{X}_{\boldsymbol \mu} =\hat{V}_L^{0 \dagger}\left( {\hat{a}_1 e^{-i \phi_1} + \hat{a}_1^\dagger e^{i \phi_1} \over \sqrt{2}}\right)\hat{V}_L^0$. To derive the third line, note that for any unit vector ${\boldsymbol v}\in \mathbb{R}^{2M}$, the operator inequality $(\hat{\boldsymbol R}^T{\boldsymbol v})^{2}\le \hat{\boldsymbol R}^T \hat{\boldsymbol R}$ holds, therefore, it follows that ${\rm Tr} \left[ (\hat\rho \otimes \ket{0}\bra{0}) \hat{X}_{\boldsymbol \mu}\right] \le \sqrt{{\rm Tr} \left[ (\hat\rho \otimes \ket{0}\bra{0}) \hat{X}_{\boldsymbol \mu}^{2}\right]} \le \sqrt{{\rm Tr} \left[ (\hat\rho \otimes \ket{0}\bra{0}) \hat{\boldsymbol R}^T \hat{\boldsymbol R}\right]} \le \sqrt{2}\sqrt{{\rm Tr} \left[ (\hat\rho \otimes \ket{0}\bra{0}) \sum_{j=1}^{N+1}(\hat{a}_{j}^{\dagger}\hat{a}_{j}+{1\over 2}) \right]}=\sqrt{2\overline{n}+N+1}$. 
%\begin{equation}
%\label{EQ3}
%\begin{aligned}
%{\rm Tr} \left[ \hat{V}_L^0 (\hat\rho \otimes \ket{0}\bra{0}) \hat{V}_L^{0 \dagger} (\hat{a}_1e^{-i \phi_1} + \hat{a}_1^\dagger e^{i \phi_1} )\right] &= \sqrt{2} {\rm Tr} \left[ (\hat\rho \otimes \ket{0}\bra{0}) \hat{X}_{\boldsymbol \mu} \right] \\
%&\leq \sqrt{2 {\rm Tr} \left[ (\hat\rho \otimes \ket{0}\bra{0}) \hat{X}_{\boldsymbol \mu}^2 \right] } \\
%&\leq \sqrt{2 {\rm Tr} \left[ (\hat\rho \otimes \ket{0}\bra{0}) (\hat{X}_{\boldsymbol \mu}^2 + \hat{P}_{\boldsymbol \mu}^2) \right] } \\
%&\leq \sqrt{ 2 ( 2\bar{n} + 1)}.
%\end{aligned}
%\end{equation}

By combining the two expressions given above, we have
$$
\begin{aligned}
{\cal M}_{\rm phase}^\alpha(\hat\rho) &\geq \frac{1}{4} I_F (\hat{U}_L^\alpha (\hat\rho \otimes \ket{0}\bra{0}) \hat{U}_L^{\alpha \dagger}, \hat{a}_1^\dagger \hat{a}_1) - {\rm Tr} \left[ \hat{U}_L^\alpha (\hat\rho \otimes \ket{0}\bra{0}) \hat{U}_L^{\alpha \dagger} \hat{a}_1^\dagger \hat{a}_1\right] \\
& \geq   |\alpha|^2 ({\cal I}_{\rm opt}(\hat\rho \otimes \ket{0}\bra{0})-1)  - 2 |\alpha| \left( \sqrt{{\cal I}_{\rm phase}^0(\hat\rho \otimes \ket{0}\bra{0}) {\cal I}_{\rm opt}(\hat\rho \otimes \ket{0}\bra{0})} + \sqrt{\bar{n} + \frac{N+1}{2}} \right) - \bar{n} \\
& \geq  |\alpha|^2 {\cal M}(\hat\rho) - 2 |\alpha| K  - \bar{n},
\end{aligned}
$$
where ${\cal M}(\hat\rho) = {\cal I}_{\rm opt}(\hat\rho \otimes \ket{0}\bra{0}) -1 > 0$ and $K =  \sqrt{{\cal I}_{\rm phase}^0(\hat\rho \otimes \ket{0}\bra{0}) {\cal I}_{\rm opt}(\hat\rho \otimes \ket{0}\bra{0})} +\sqrt{\bar{n} + \frac{N+1}{2}} >0$. Therefore we can always choose sufficiently large $|\alpha| > \frac{K + \sqrt{K^2 + {\cal M}(\hat\rho)^2 \bar{n} }}{{\cal M}(\hat\rho)}$ and linear optical unitary $\hat{U}_L^\alpha$ to achieve ${\cal M}_{\rm phase}^\alpha(\hat\rho)>0$. Furthermore, for a large value of $\alpha$, we can always have 
$$
\lim_{|\alpha| \rightarrow \infty} \frac{{\cal M}_{\rm phase}^\alpha (\hat\rho)}{|\alpha|^2} \geq {\cal M}(\hat\rho).
$$
Meanwhile, we can also observe that
$$
{\cal M}_{\rm phase}^\alpha (\hat\rho) \leq {\cal I}_{\rm phase}^\alpha(\hat\rho \otimes \ket{0}\bra{0} ) \leq \left[ \sqrt{{\cal I}_{\rm phase}^0(\hat\rho \otimes \ket{0}\bra{0} )} + |\alpha| \sqrt{{\cal I}_{\rm opt}(\hat\rho \otimes \ket{0}\bra{0} )}\right]^2,
$$
then $$\lim_{|\alpha| \rightarrow \infty} \frac{{\cal M}_{\rm phase}^\alpha (\hat\rho)}{|\alpha|^2} \leq {\cal I}_{\rm opt}(\hat\rho \otimes \ket{0}\bra{0} ) = {\cal M}(\hat\rho) + 1.$$
\end{proof}

\subsection{Proof of Theorem~4}
\begin{proof}
We first note that $\hat\rho_0= \hat{W}_L \hat\rho \hat{W}_L^\dagger$ satisfying ${\rm Tr} \hat\rho_0 \hat{\boldsymbol R} = 0$ can be obtained by a linear optical unitary $\hat{W}_L = \prod_{n=1}^N \hat{D}_n(\beta_n)$, where $\beta_n = -{\rm Tr} \left[ \hat\rho (\hat{x}_n  + i \hat{p}_n) /\sqrt{2} \right]$. Therefore, optimization over linear optical unitaries is equivalent between starting with $\hat\rho_0$ and $\hat\rho$. We also note that the metrological power ${\cal M}$ with respect to displacement estimation is invariant under a linear optical unitary, i.e. ${\cal M}(\hat\rho_0) = {\cal M}(\hat\rho)$. We then introduce the big/small-O and $\Theta$-notations to describe the leading order of a real-valued function $f(x)$. We refer $f(x) = {\cal O}(x^n)$ when $\limsup_{x\rightarrow \infty} f(x) / x^n < \infty $ and $f(x) = o(x^n)$ when $\limsup_{x\rightarrow \infty} f(x) / x^n =0 $. We also say $f(x) = \Theta(x^n)$ when $f(x)$ scales by $x^n$ by means of $\exists k_1, k_2 > 0$ and $\exists M>0$ such that $\forall x > M$, $k_1 x^n \leq f(x) \leq k_2 x^n$.\\

%For a real valued function $f(x)$, we refer $f(x) = {\cal O}(x^n)$ when $\displaystyle \limsup_{x \rightarrow \infty} |f(x)|/x^n < \infty$ and $f(x) = o(x^n)$ when $\displaystyle \lim_{x \rightarrow \infty} f(x)/x^n = 0$.
The {``\it if"} part can be proven as follows:
Note that if ${\cal M}_{\rm phase}^0(\hat\rho_0) = {\Theta } (\bar{n}_{\hat\rho_0}^k)$ with $k\geq2$, then the HS can obviously be obtained by $I_F(\hat{U}_L^0 (\hat\rho_0 \otimes \ket{0}\bra{0}) \hat{U}_L^{0 \dagger}, \hat{a}_1^\dagger \hat{a}_1) \geq 4 {\cal M}_{\rm phase}^0(\hat\rho_0) = \Theta(\bar{n}_{\hat\rho_0}^k)$ with $k\geq2$ because $\hat{U}_L^0$ does not change the mean photon number in the system. On the other hand, suppose that ${\cal M}(\hat\rho_0) = \Theta (\bar{n}_{\hat\rho_0})$  which is equivalent to ${\cal I}_{\rm opt} (\hat\rho) = {\cal I}_{\rm opt} (\hat\rho_0) = \Theta (\bar{n}_{\hat\rho_0})$  and ${\cal M}_{\rm phase}^0 = \Theta(\bar{n}_{\hat\rho_0}^k)$ with $k<2$. Then we can choose the optimal linear optical unitary $\hat{U}_L^\alpha = \hat{D}_1(|\alpha|) \hat{V}_L^0$ to get ${\cal I}_{\rm opt} (\hat\rho_0 \otimes \ket{0}\bra{0}) = (1/2) I_F( \hat{V}_L^0 (\hat\rho_0 \otimes \ket{0}\bra{0}) \hat{V}_L^{0 \dagger} , \hat{x}_1)$ and $|\alpha| = \sqrt{\kappa \bar{n}_{\hat\rho_0}}$ with a sufficiently large constant $\kappa$.
% \geq \frac{{\cal I}_{\rm phase}^0(\hat\rho \otimes \ket{0}\bra{0})} {\bar{n} {\cal I}_{\rm opt}(\hat\rho \otimes \ket{0}\bra{0})}$.
In this case, we have
$$
\begin{aligned}
I_F( \hat{U}_L^\alpha (\hat\rho_0 \otimes \ket{0}\bra{0}) \hat{U}_L^{\alpha \dagger}, \hat{a}_1^\dagger \hat{a}_1) &\geq \left[ \sqrt{{\cal I}_{\rm phase}^0 (\hat\rho_0 \otimes \ket{0}\bra{0})} - |\alpha| \sqrt{{\cal I}_{\rm opt} (\hat\rho_0 \otimes \ket{0}\bra{0})}\right]^2 \\
&\approx \kappa \bar{n}_{\hat\rho_0} {\cal I}_{\rm opt} (\hat\rho_0 \otimes \ket{0}\bra{0})\\
%&= k \bar{n} {\cal I}_{\rm opt} (\hat\rho \otimes \ket{0}\bra{0}) \\
%& = k \bar{n} \Theta (\bar{n})\\
& = \Theta (\bar{n}_{\hat\rho_0}^2)
\end{aligned}
$$
by using Proposition~\ref{Prop2}.
At the same time, we have 
$$
\bar{n}_{\hat\sigma} = \sum_{n=1}^{N+1} {\rm Tr} [\hat{U}_L^\alpha (\hat\rho_0 \otimes \ket{0} \bra{0}) \hat{U}_L^{\alpha \dagger} \hat{a}_n^\dagger \hat{a}_n]  = \bar{n}_{\hat\rho_0} + |\alpha|^2 = \Theta (\bar{n}_{\hat\rho_0}).
$$
Therefore, it is possible to reach the HS 
$$ 
I_F( \hat{\sigma}, \hat{a}_1^\dagger \hat{a}_1) = \Theta (\bar{n}_{\hat\sigma}^2),
$$
where $\hat\sigma = \hat{U}_L^\alpha (\hat\rho_0 \otimes \ket{0}\bra{0}) \hat{U}_L^{\alpha \dagger}$ is obtained by using the linear optical unitary $\hat{U}_L^\alpha$ in addition to the vacuum ancilla $\ket{0}$.\\

The {``\it only if"} part can be proven as follows:
Suppose that the scaling of the parameters are given by ${\cal M}_{\rm phase}^0(\hat\rho_0) =  \Theta (\bar{n}_{\hat\rho_0}^k)$, ${\cal M}_{\rm opt}(\hat\rho_0) = \Theta (\bar{n}_{\hat\rho_0}^l)$, and $|\alpha|^2 = \Theta (\bar{n}_{\hat\rho_0}^m)$.
By using Proposition~\ref{Prop2}, we then have 
$$
I_F(\hat\sigma, \hat{a}^\dagger_1\hat{a}_1) = {\cal O}(\bar{n}_{\hat\rho_0}^k) + {\cal O}(\bar{n}_{\hat\rho_0}^{m+l}) + {\cal O}(\bar{n}_{\hat\rho_0}^{(k+l+m)/2})  = {\cal O}(\bar{n}_{\hat\rho_0}^p),
$$
where $\hat\sigma = \hat{U}_L^\alpha (\hat\rho_0 \otimes \ket{0}\bra{0}) \hat{U}_L^{\alpha \dagger}$ and $p = {{\rm max} \{ k, m+l, (k+l+m)/2\} }$.
Meanwhile, the mean photon number after acting the linear optical unitary $\hat{U}_L^\alpha$ scales by
$$
\bar{n}_{\hat\sigma} = \sum_{n=1}^{N+1} {\rm Tr} [\hat{U}_L^\alpha (\hat\rho_0 \otimes \ket{0} \bra{0}) \hat{U}_L^{\alpha \dagger} \hat{a}_n^\dagger \hat{a}_n]  = \bar{n}_{\hat\rho_0} + |\alpha|^2 = \Theta (\bar{n}_{\hat\rho_0}^q),
$$
where $q = {\rm max} \{ m, 1 \}$. Then we calculate the scaling behavior of $I_F(\hat\sigma, \hat{a}^\dagger_1\hat{a}_1)$ by $\bar{n}_{\hat\sigma}$ as
$$
I_F(\hat\sigma, \hat{a}^\dagger_1\hat{a}_1) = {\cal O} (\bar{n}_{\hat\rho_0}^p) = {\cal O}(\bar{n}_{\hat\sigma}^r) = o(\bar{n}_{\hat\sigma}^2)
$$
with $r<2$ when $k<2$ and $l<1$, regardless of the choice of $m$. In this case, therefore, it is impossible to reach the HS by using linear optical unitary in addition to the vacuum ancilla.\\

Finally, we show that ${\cal M}(\hat\rho) = \Theta(\bar{n}_{\hat\rho})$ is the sufficient condition to reach the HS by noting that
$$
{\cal M}(\hat\rho_0) = {\cal M}(\hat\rho) = \Theta(\bar{n}_{\hat\rho}) \geq \Theta (\bar{n}_{\hat\rho_0})
$$
since $\bar{n}_{\hat\rho} = \bar{n}_{\hat\rho_0} + |\beta|^2 \geq \bar{n}_{\hat\rho_0}$ where $|\beta|^2 = \sum_{n=1}^N |\beta_n|^2$ for $\hat{W}_L = \prod_{n=1}^N \hat{D}_n(\beta_n)$ defined above.

\end{proof}


\begin{thebibliography}{99}
\bibitem{noteadded} The author of the present Letter and the authors of Ref.~\cite{YadinX} have independently developed the two closely related methods and found each other's work before concurrent submission to a journal.
\bibitem{note1} This can be obtained by using the following property of the QFI: $ \bigg| \sqrt{I_F(\hat\rho, \hat{A})} - \sqrt{I_F(\hat\rho, \hat{B})} \bigg| \leq \sqrt{I_F(\hat\rho, \hat{A}+\hat{B})} \leq \sqrt{I_F(\hat\rho, \hat{A})} + \sqrt{I_F(\hat\rho, \hat{B})}$. See Ref.~ \cite{Suppl} for a detailed proof.
\end{thebibliography}
\end{document}